\documentclass{jpp}

\newcommand{\be}{\begin{equation}}
\newcommand{\ee}{\end{equation}}
\let\proof\relax 
\usepackage[psamsfonts]{amsfonts}
\usepackage{amsmath, amsthm, amssymb}
\usepackage{multirow}
\usepackage{enumerate}
\usepackage{setspace}
\usepackage[dvips]{graphicx}
\usepackage{natbib}

\newcommand{\derv}[1]{\frac{\partial}{\partial #1}}

\newcommand{\deriv}[2]{\frac{\partial #1}{\partial #2}}

\newcommand{\beqn}{\begin{equation}}
\newcommand{\eeqn}{\end{equation}}
\newcommand{\beqnar}{\begin{eqnarray}}
\newcommand{\eeqnar}{\end{eqnarray}}

\newtheorem{theorem}{Theorem}[section]
\newtheorem{proposition}[theorem]{Proposition}
\newenvironment{example}[1][Example]{\begin{trivlist}
\item[\hskip \labelsep {\bfseries #1}]}{\end{trivlist}}
\newenvironment{remark}[1][Remark]{\begin{trivlist}
\item[\hskip \labelsep {\bfseries #1}]}{\end{trivlist}}

\ifCUPmtlplainloaded \else
  \checkfont{eurm10}
  \iffontfound
    \IfFileExists{upmath.sty}
      {\typeout{^^JFound AMS Euler Roman fonts on the system,
                   using the 'upmath' package.^^J}%
       \usepackage{upmath}}
      {\typeout{^^JFound AMS Euler Roman fonts on the system, but you
                   dont seem to have the}%
       \typeout{'upmath' package installed. JPP.cls can take advantage
                 of these fonts, if you use 'upmath' package.^^J}%
      }
  \else
  \fi
\fi

\ifCUPmtlplainloaded \else
  \checkfont{msam10}
  \iffontfound
    \IfFileExists{amssymb.sty}
      {\typeout{^^JFound AMS Symbol fonts on the system, using the
                'amssymb' package.^^J}%
       \usepackage{amssymb}%
         \let\leq=\leqslant
         
      }{}
  \fi
\fi

\ifCUPmtlplainloaded \else
  \IfFileExists{amsbsy.sty}
    {\typeout{^^JFound the 'amsbsy' package on the system, using it.^^J}%
     \usepackage{amsbsy}}
    {\providecommand\boldsymbol[1]{\mbox{\boldmath $##1$}}}
\fi

\newsavebox{\astrutbox}
\sbox{\astrutbox}{\rule[-5pt]{0pt}{20pt}}

\title[Multi-Symplectic Magnetohydrodynamics]
{Multi-Symplectic Magnetohydrodynamics}

\author[G. M. Webb, J. F. McKenzie and G. P. Zank]%
{G.\ns M.\ns W\ls E\ls B\ls B$^1$%
 \thanks{Email address for correspondence: gmw0002@uah.edu},\ns
J.\ls F.\ns M\ls c\ls K\ls E\ls N\ls Z\ls I\ls E$^{1,3}$\break
 \and G.\ns P.\ns Z\ls A\ls N\ls K$^{1,2}$}

\affiliation{$^1$Center for Space Plasma and Aeronomic Research, 
The University of Alabama in Huntsville, 
Huntsville AL 35805, USA\\[\affilskip]
$^2$Department of Space Science, The University of
Alabama in Huntsville, Huntsville AL 35805, USA\\
$^3$Department of Mathematics and Statistics, 
Durban University of Technology,\\
Steve Biko Campus, Durban South Africa, and School of Mathematical Sciences,University of Kwa-Zulu Natal, Durban South Africa}

\begin{document}


\maketitle

\begin{abstract}
A multi-symplectic formulation of ideal magnetohydrodynamics (MHD) 
is developed based on a Clebsch variable  
 variational 
principle  in which the Lagrangian 
consists of the kinetic minus the potential energy of the MHD fluid
modified by constraints using 
Lagrange multipliers, that ensure mass conservation, entropy advection with
the flow, the Lin constraint and Faraday's equation (i.e the magnetic flux
is Lie dragged with the flow). The analysis is also carried out 
using the magnetic vector potential $\tilde{\bf A}$ where 
$\alpha=\tilde{\bf A}{\bf\cdot}d{\bf x}$ is Lie dragged with the flow, and
${\bf B}=\nabla\times\tilde{\bf A}$. 
The multi-symplectic conservation laws give rise to 
the Eulerian momentum and energy conservation laws. The symplecticity 
or structural conservation laws for the multi-symplectic system corresponds 
to the conservation of phase space. It corresponds to 
taking derivatives of the momentum and energy conservation laws 
and combining them to produce $n(n-1)/2$ extra conservation 
laws, where $n$ is the number of independent variables. Noether's 
theorem for the multi-symplectic MHD system is derived, including 
the case of non-Cartesian space coordinates, where the metric plays a
role in the equations. 
\end{abstract}


\maketitle


\section{Introduction}
Multi-symplectic equations for
Hamiltonian systems with two or more
independent variables $x^\alpha$ have been developed 
as a useful extension of
Hamiltonian systems with one evolution variable $t$.
This development has connections with dual
variational formulations
of travelling wave problems (e.g. \cite{Bridges92}), and is useful in 
numerical schemes for
Hamiltonian systems. Bridges and co-workers used the multi-symplectic
approach to study
 linear and nonlinear wave propagation, 
 generalizations of 
wave action, wave modulation 
theory, and wave stability problems (\cite{Bridges97a,Bridges97b}). 
 \cite{Reich00} and  \cite{Bridges06} develop multi-symplectic difference schemes. 
Multi-symplectic Hamiltonian systems 
have been studied by 
\cite{Marsden99} and  \cite{Bridges05}.  
\cite{Webb07,Webb08,Webb14d} discuss 
travelling waves in multi-fluid plasmas using a multi-symplectic formulation 
analogous to \cite{Bridges92} paper on travelling water waves. \cite{Holm98}
give an overview of Hamiltonian systems, semi-direct 
product Lie algebras and Euler-Poincar\'e equations.  

\cite{Cotter07} developed a multi-symplectic, Euler-Poincar\'e
formulation of fluid mechanics.  They showed  
that  multi-symplectic ideal fluid mechanics type systems are 
 related to  Clebsch variable formulations  
 in which the Lagrange multipliers 
play the role of canonically conjugate momenta to the 
constrained variables {\bf (\cite{Zakharov97}, \cite{Morrison98})}.  
Thus, the Clebsch variable formulation involves a momentum map.  
 As a part of the paper, we give 
a brief introduction to  
 multi-symplectic systems (Sections 3 and 4),  
based on the work of \cite{Hydon05} and \cite{Cotter07}  
(see also \cite{Brio10}). In multi-symplectic  
 Hamiltonian systems, both the space and the time variables 
can be thought of as evolution variables.   
In this paper we obtain multi-symplectic equations for ideal gas dynamics 
and MHD, based on the Clebsch variables formulation. 
 The energy and momentum conservation laws  
for gas dynamics and MHD are obtained from the multi-symplectic formalism.
 We also discuss and give examples of the symplecticity  or structural 
conservation laws, which are obtained by combining derivatives of the
energy and momentum conservation laws. The symplecticity conservation 
law thus impose higher order derivative constraints on the energy 
and momentum conservation laws. 

The present paper is related to recent work by 
\cite{Webb14a,Webb14b,Webb14c} and \cite{WebbMace14} 
on conservation laws,  Hamiltonian and symmetry group approaches to the 
MHD equations. In \cite{Webb14a} advected invariants in ideal fluid 
mechanics and MHD were obtained by using Lie dragging techniques 
(see also \cite{Tur93}). In particular helicity  
conservation laws were discussed 
(i.e. fluid helicity conservation in ideal fluids (e.g. \cite{Moffatt69}), 
cross helicity and magnetic helicity in MHD (\cite{Woltjer58}).  
 \cite{Berger84} investigated magnetic helicity 
and relative magnetic helicity (see also \cite{Finn85,Finn88}). 
\cite{Berger90}, \cite{Ruzmaikin94} and 
\cite{Akhmetiev95} investigated higher order MHD knot invariants 
known as Sato-Levine invariants, which can be described using Seifert surfaces
(e.g. used to describe the Whitehead link and Borromean rings). 
 For special ideal MHD flows in which the magnetic helicity density 
${\bf A}{\bf\cdot}{\bf B}$ vanishes, there is a higher order magnetic 
topological charge, namely the Godbillon Vey invariant (e.g. 
\cite{Tur93},  
 \cite{Webb14a,Webb14c}) which also describes the
magnetic field topology (i.e. magnetic helicity is not the only quantity 
describing magnetic topology). For the Godbillon Vey invariant, the  gauge 
of the magnetic vector potential ${\bf A}$ is chosen so that the one-form 
$\alpha={\bf A}{\bf\cdot} d{\bf x}$ is Lie dragged by the 
background plasma flow. \cite{Webb14b} derive MHD conservation laws 
using Noether's first and second theorems ({\bf see \cite{Hydon11} for 
an insightful treatment of Noether's second theorem}). In the most general case, the 
fluid helicity and cross helicity conservation laws are nonlocal as they 
involve Clebsch potentials, which are nonlocal variables (e.g. 
\cite{Webb14a,Webb14b}).  
 A brief synopsis of this work is given by \cite{Webb14c}. 
{\bf \cite{WebbMace14}
use Noether's second theorem and fluid relabelling symmetries 
to derive conservation laws in MHD using the approach of \cite{Hydon11}. 
They obtain a new potential vorticity type conservation law for MHD using a 
non-field aligned fluid relabelling symmetry of the equations.} 

The main aim of the present paper is to provide a multi-symplectic 
formulation of the MHD equations by using Clebsch variables. 
 
In Section 2 the basic ideal MHD equations are given. 
{\bf In Section 2, we give a simple example of the multi-symplectic 
form of the ideal gas dynamic equations in one Cartesian space coordinate. 
We give a recipe for constructing the multi-symplectic
form of the equations and describe the symplecticity or structural 
conservation law for 1D gas dynamics. By using Clebsch variables, one 
is in effect increasing the number of dependent variables describing the 
system. To ensure that the system of equations describes the original Eulerian 
fluid equations, extra constraints or conservation laws (the symplecticity 
conservation laws) ensures that original system of equations is obtained. 
It turns out that the symplecticity conservation equations can be expressed 
in terms of two-dimensional Jacobians of the dependent variables 
and the independent space and time variables. 
 The basic theory behind this approach is described 
in Section 4.} 

In Section 3, we give an introduction to Clebsch variables 
and momentum maps following the analysis of \cite{Cotter07}.  
We introduce the Clebsch variable description of MHD 
(e.g. \cite{Zakharov97}, {\bf \cite{Morrison98}, \cite{Holm83a,Holm83b},}   
 \cite{Webb14a,Webb14b}). We give two different 
formulations of the MHD variational principle. In the first formulation, 
 Faraday's law for the magnetic field induction ${\bf B}$ is included 
as a constraint. The constraints of mass continuity, entropy advection, 
and the Lin constraint are incorporated in the variational principle 
by means of Lagrange multipliers. In a second variational formulation, 
the Faraday equation constraint is replaced by the condition that 
the magnetic vector potential 1-form: 
$\alpha=\tilde{\bf A}{\bf\cdot}d{\bf x}$ 
is Lie dragged by the flow, where ${\bf B}=\nabla\times\tilde{\bf A}$
(e.g. \cite{Tur93}, {\bf \cite{Gordin87}, \cite{Padhye96a,Padhye96b}} 
and \cite{Webb14a}). In the case 
$\nabla{\bf\cdot}{\bf B}=0$, and  ${\bf B}=\nabla\times\tilde{\bf A}$, the 
condition that the one-form $\alpha=\tilde{\bf A}{\bf\cdot}d{\bf x}$ 
is Lie dragged by the flow is equivalent to Faraday's equation. 
{\bf We include an extra ${\bf u}\nabla{\bf\cdot}{\bf B}$ 
term in Faraday's equation for mathematical reasons. This allows 
one to obtain a mapping between the multi-symplectic MHD 
equations using the advected ${\bf A}$ formalism and the advected 
magnetic flux version of Faraday's law (see Proposition (\ref{prop5.3}) and 
(\ref{eq:con7}) in the conclusions). 
It is also useful to include the ${\bf u}\nabla{\bf\cdot}{\bf B}$ 
in numerical MHD in order to minimize numerically 
generated $\nabla{\bf\cdot}{\bf B}$ in numerical MHD codes  by advecting 
the numerically generated $\nabla{\bf\cdot}{\bf B}\neq 0$ out of the computational 
domain
 (e.g. \cite{Powell99}, \cite{WebbPogorelovZank10}). By setting 
$\nabla{\bf\cdot}{\bf B}=0$ after all calculations are done gives rise 
to correct physical results, but allows one to retain a mathematically 
more useful form of the equations.} 
In general, Faraday's equation is equivalent to the condition that 
the magnetic field flux 2-form $\beta={\bf B}{\bf\cdot}d{\bf S}$ is Lie 
dragged with the flow. If $\nabla{\bf\cdot}{\bf B}=0$ then $\beta=d\alpha$ 
where $d\alpha$ is the exterior derivative of the 1-form $\alpha$. 

Section 4 presents an overview of multi-symplectic Hamiltonian systems 
based in part on the work of  \cite{Hydon05}. We discuss the symplecticity 
conservation laws, which are related to the conservation of phase space in 
multi-symplectic systems, in which the generalized phase space volume element 
$\kappa^\alpha=(1/2){\sf K}^\alpha_{ij} dz^i\wedge dz^j$ is given by the 
exterior derivative of a one-form $\omega^\alpha$. This requirement implies 
$d\kappa^\alpha=dd\omega^\alpha=0$ (by the Poincar\'e
 Lemma), which implies the conservation of the 
phase space element $\kappa^\alpha$ where $z^i$ are the canonical variables. 
We also discuss Noether's first theorem for multi-symplectic systems.  
We  describe the form of  multi-symplectic  systems for the case 
of non-Cartesian spatial coordinates where the spatial metric 
plays a role in the equations (see also \cite{Bridges10}). 

Section 5 develops the multi-symplectic MHD evolution equations. 
The analysis is similar to that of \cite{Cotter07} 
where they develop the multi-symplectic approach to Hamiltonian hydrodynamic 
type systems.

Section 6 concludes with a summary and discussion.  
\section{The Model}

The magnetohydrodynamic equations can be written in the form:
\beqnar
&&\deriv{\rho}{t}+\nabla{\bf\cdot} (\rho{\bf u})=0, \label{eq:2.1}\\
&&\derv{t}(\rho{\bf u})+\nabla{\bf\cdot}\left[\rho {\bf u}{\bf u}
+\left(p+\frac{B^2}{2\mu}\right){\bf I} -\frac{\bf B B}{\mu}\right]=0,
\label{eq:2.2}\\
&&\deriv{S}{t}+{\bf u}{\bf\cdot}\nabla S=0, 
\label{eq:2.3}\\
&&\deriv{\bf B}{t}-\nabla\times\left({\bf u}\times {\bf B}\right)
+{\bf u}\nabla{\bf\cdot}{\bf B}=0, \label{eq:2.4}
\eeqnar
where $\rho$, ${\bf u}$, $p$,
$S$ and ${\bf B}$ correspond to the gas density, fluid velocity, pressure,
specific entropy, and magnetic induction ${\bf B}$ respectively, and
${\bf I}$ is the unit $3\times 3$ dyadic.
The gas pressure $p=p(\rho,S)$ is a function of the density $\rho$ and
entropy $S$, and $\mu$ is the magnetic permeability. 
Equations (\ref{eq:2.1})-(\ref{eq:2.2}) are the mass and  momentum  
conservation laws, (\ref{eq:2.3}) is the entropy advection equation 
and (\ref{eq:2.4}) is Faraday's equation in the MHD limit. 

In classical MHD, (\ref{eq:2.1})-(\ref{eq:2.4}) are supplemented by
Gauss' law:
\beqn
\nabla{\bf\cdot}{\bf B}=0, \label{eq:2.5}
\eeqn
which implies the non-existence of magnetic monopoles.

It is useful to keep in mind the first law of thermodynamics:
\beqn
TdS=dQ=dU+pdV\quad\hbox{where}\quad V=\frac{1}{\rho}, \label{eq:2.7}
\eeqn
where $U$ is the internal energy per unit mass and $V=1/\rho$ is the specific
volume. Using the internal energy per unit volume $\varepsilon=\rho U$
instead of $U$, (\ref{eq:2.7}) may be written as:
\beqn
TdS=\frac{1}{\rho}\left(d\varepsilon-hd\rho\right)\quad\hbox{where}\quad
h=\frac{\varepsilon+p}{\rho}, \label{eq:2.8}
\eeqn
is the enthalpy of the gas. Assuming $\varepsilon=\varepsilon(\rho,S)$,
 (\ref{eq:2.8}) gives  the formulae:
\beqn
 \rho T=\varepsilon_S, \quad h=\varepsilon_\rho,
\quad p=\rho\varepsilon_{\rho}-\varepsilon, \label{eq:2.9}
\eeqn
relating the temperature $T$, enthalpy $h$ and pressure $p$ to the
internal energy density $\varepsilon(\rho,S)$. From (\ref{eq:2.8})
we obtain:
\beqn
TdS=dh-\frac{1}{\rho}dp\quad \hbox{and}\quad
-\frac{1}{\rho}\nabla p=T\nabla S-\nabla h, \label{eq:2.10}
\eeqn
which is useful in the further analysis of the momentum equation for the
system.

\subsection{Multi-Symplectic Example}
{\bf Multi-symplectic systems are a generalization of Hamiltonian systems.
Consider the equations of 1D gas dynamics involving the variables 
${\bf z}=(u,\rho,S,\beta,\phi)^T$ where 
\beqn
u=\deriv{\phi}{x}-\frac{\beta}{\rho} \deriv{S}{x}, \label{eq:2.1a}
\eeqn
is the Clebsch representation for the fluid velocity $u\equiv u^x$ directed 
along the $x$-axis. The governing equations have the form (e.g. 
\cite{Zakharov97}, 
 \cite{Morrison98}):
\begin{align}
&\rho_t+(\rho u)_x=0,\quad S_t+u S_x=0, \label{eq:2.2a}\\
&\phi_t+u\phi_x=\frac{1}{2} u^2-h, \quad \beta_t+(\beta u)_x=-\rho T, 
\label{eq:2.3a}
\end{align}
where $\varepsilon(\rho,S)$ is the the internal energy per unit volume of the gas. 
 Here $\rho$, $T$, $p$, $h$, $S$ denote the density, temperature, pressure, enthalpy 
and entropy 
of the gas,  
 where $h=\varepsilon_\rho$, $\rho T=\varepsilon_S$
and $p=\rho\varepsilon_\rho-\varepsilon$.  The time evolution of a functional 
$F$ of the physical variables satisfies Hamilton's equations:  
\beqn 
F_t=\left\{F,H\right\}=\int\ dx\biggl(\frac{\delta F}{\delta\rho}\frac{\delta H}{\delta\phi}
-\frac{\delta F}{\delta\phi}\frac{\delta H}{\delta\rho}
+\frac{\delta F}{\delta S}\frac{\delta H}{\delta\beta}
-\frac{\delta F}{\delta\beta}\frac{\delta H}{\delta S}\biggr), \label{eq:2.4a}
\eeqn
where $\left\{F,H\right\}$ is the canonical Poisson bracket and
\beqn 
H=\int dx\left(\frac{1}{2}\rho u^2+\varepsilon(\rho,S)\right) \label{eq:2.5a}
\eeqn 
is the Hamiltonian functional in which $u$ is given by (\ref{eq:2.1a}). 

Equations (\ref{eq:2.1a})-(\ref{eq:2.3a}) can  be written in the 
multi-symplectic form: 
\beqn
\left({\sf K}^0\derv{t}+{\sf K}^1\derv{x}\right){\bf z}=\frac{\delta M}{\delta {\bf z}}, 
\label{eq:2.6a}
\eeqn
where
\beqn
M=-\int\ dx\ \ell=-\int\ dx\left(\frac{1}{2}\rho u^2-\varepsilon(\rho,S)
\right), \label{eq:2.7a}
\eeqn
is the multi-symplectic Hamiltonian (note $\ell$ is the Lagrange density of the fluid). 
The skew symmetric matrices ${\sf K}^0$ and ${\sf K}^1$ are given by:
\beqn
{\sf K}^0=\left(\begin{array}{ccccc}
0&0&0&0&0\\
0&0&0&0&-1\\
0&0&0&-1&0\\
0&0&1&0&0\\
0&1&0&0&0\\
\end{array}
\right),\quad 
{\sf K}^1=\left(\begin{array}{ccccc}
0&0&\beta&0&-\rho\\
0&0&0&0&-u\\
-\beta&0&0&-u&0\\
0&0&u&0&0\\
\rho&u&0&0&0
\end{array}
\right).\label{eq:2.8a}
\eeqn
In the (\ref{eq:2.6a}) there are two skew symmetric matrices ${\sf K}^0$
and ${\sf K}^{1}$, which shows that both space and time can be thought of as 
evolution variables. 
\cite{Bridges92, Bridges97a, Bridges97b, Bridges06, Bridges10} 
has championed the use of multi-symplectic methods for 
nonlinear wave problems, both in numerical methods for solving nonlinear wave equations, 
wave action and wave stability problems, and Hamiltonian bifurcation theory
 (i.e wave stability theory depending on a bifurcation
parameter). \cite{Cotter07} develop multi-symplectic approaches to 
incompressible fluid dynamics, and other systems.

Below we illustrate the recipe for obtaining the multi-symplectic form 
(\ref{eq:2.6a})-(\ref{eq:2.8a}) based on the results of 
Section 4. 

The constrained Lagrangian associated with the Clebsch representation in the 
present example has the form:
\begin{equation}
L=\frac{1}{2}\rho u^2-\varepsilon(\rho,S)+L^\alpha_{z^s} \deriv{z^s}{x^\alpha}, 
\label{eq:2.9a}
\end{equation}
where the term:
\begin{equation}
L^\alpha_{z^s} \deriv{z^s}{x^\alpha}
=\phi\left[\deriv{\rho}{t}+\derv{x}(\rho u)\right] 
+\beta\left(\deriv{S}{t}+u \deriv{S}{x}\right), \label{eq:2.10a}
\end{equation}
contains the Lagrangian constraints associated with the mass 
continuity equation 
and the entropy advection equation, and ${\bf z}=(u,\rho,S,\beta,\phi)^T$. In 
(\ref{eq:2.10a}) we identify 
\begin{align}
&L^0_\rho=\phi,\quad L^1_\rho=\phi u,\quad L^1_u=\phi \rho, \nonumber\\
&L^0_S=\beta,\quad L^1_S=\beta u. \label{eq:2.11a}
\end{align}
The one-forms:
\begin{equation}
\omega^\alpha= L^\alpha_{z^s} dz^s, \quad (\alpha=0,1), \label{eq:2.12a}
\end{equation}
using (\ref{eq:2.11a}) are given by:
\begin{equation}
\omega^0=\phi d\rho+\beta dS,\quad \omega^1=u(\phi d\rho+\beta dS)+\phi\rho\ du. \label{eq:2.13a}
\end{equation}
The exterior derivatives of the 1-forms (\ref{eq:2.13a}) are:
\begin{align}
&d\omega^0=d\phi\wedge d\rho+d\beta\wedge dS\equiv 
\frac{1}{2} {\sf K}^0_{\alpha\beta} dz^\alpha\wedge dz^\beta, \nonumber\\
&d\omega^1=du\wedge\left(\beta dS-\rho d\phi\right)
+u(d\phi\wedge d\rho+d\beta\wedge dS)
\equiv \frac{1}{2} {\sf K}^1_{\alpha\beta} dz^\alpha\wedge dz^\beta. 
\label{eq:2.14a}
\end{align}
The multi-symplectic matrices ${\sf K}^0_{\alpha\beta}$ 
and ${\sf K}^1_{\alpha\beta}$ given in (\ref{eq:2.8a}) can be 
determined from (\ref{eq:2.14a}).

The multi-symplectic formalism can be used to obtain conservation laws, 
by using the properties of the differential forms $\omega^\alpha$ defining 
the system (see Section 4). Conservation laws can also be obtained by using 
the multi-symplectic version of Noether's theorem. 
In Appendix A, we show how the multi-symplectic approach gives rise to the 
energy and momentum conservation equations of 1D gas dynamics, namely:
\begin{align}
&G_0=\derv{t}\left[\frac{1}{2}\rho u^2+\varepsilon(\rho,S)\right] 
+\derv{x}\left[\rho u\left(\frac{1}{2}u^2+h\right)\right]=0, \label{eq:2.15a}\\
&G_1=-\left[\derv{t}\left(\rho u^2\right)
+\derv{x}\left(p+\rho u^2\right)\right]=0. 
\label{eq:2.16a}
\end{align}
The multi-symplectic approach also gives rise to the symplecticity 
or structural conservation laws. For the case of 1D gas dynamics, 
there is only one structural conservation law, namely:
\begin{equation}
\deriv{D}{t}+\deriv{F}{x}=0, \label{eq:2.17a}
\end{equation}
where
\begin{align}
D=&\frac{\partial(\phi,\rho)}{\partial(t,x)}
+\frac{\partial(\beta,S)}{\partial(t,x)}, \label{eq:2.18a}\\
F=&\rho\frac{\partial(\phi,u)}{\partial(t,x)}+u \frac{\partial(\phi,\rho)}{\partial(t,x)}
+\beta \frac{\partial(u,S)}{\partial(t,x)}
+u \frac{\partial(\beta,S)}{\partial(t,x)}\nonumber\\
\equiv&\frac{\partial(u\phi,\rho)}{\partial(t,x)}
+\frac{\partial(\rho\phi,u)}{\partial(t,x)}
+\frac{\partial(u\beta,S)}{\partial(t,x)}
, \label{eq:2.19a}
\end{align}
where $\partial(\phi,\psi)/\partial(t,x)=\phi_t\psi_x-\phi_x\psi_t$ is the 
Jacobian of $\phi$ and $\psi$ with respect to $t$ and $x$. 
The symplecticity conservation law (\ref{eq:2.17a}) 
corresponds to the conservation law:
\begin{equation}
D_x G_0-D_t G_1=0, \label{eq:2.20a}
\end{equation}
where $D_x\equiv \partial/\partial x$, $D_t\equiv\partial/\partial t$ 
and $G_0=0$ and $G_1=0$ are the energy and momentum conservation equations 
(\ref{eq:2.15a}) and (\ref{eq:2.16a}) written in terms of the Clebsch potentials. 

The main point here, is that the multi-symplectic structure is 
determined by the fundamental one-forms $\omega^\alpha$ ($\alpha=0,1$) 
and the Hamiltonian functional $M$. 
The theory for this is outlined in Sections 3 and 4, and 
is used to determine the multi-symplectic structure 
of the MHD equations in Section 5. We give the generalization of the 
symplecticity conservation law in the general case for MHD and gas dynamics 
in Section 5.}  
\section{Hamiltonian Approach and Clebsch Variables}
In this section we first give a synopsis of the Clebsch variational principle 
and the inverse map, discussed in more detail in \cite{Cotter07}, 
who show how Clebsch type variational principles give rise to the 
momentum map (Section 3.1).

In Section 3.2 we describe a constrained variational 
principle for MHD using Lagrange multipliers to enforce the constraints 
of mass conservation; the entropy advection equation; Faraday's 
equation and the so-called Lin constraint describing in part, the vorticity
of the flow (i.e. Kelvin's theorem). This leads to Hamilton's canonical 
equations in terms of Clebsch potentials (\cite{Zakharov97}, 
\cite{Morrison98}). {\bf \cite{Morrison80,Morrison82a}, 
and \cite{Morrison82} used the Clebsch variable formulation of MHD 
to derive the non-canonical Poisson bracket for MHD, by transforming 
the variational derivatives with respect to the Clebsch variables
to their corresponding form in terms of Eulerian physical variables. 
 Taking the variational derivative of the action with respect to 
the fluid velocity ${\bf u}$,  the Clebsch 
 variational principle gives a representation for the 
fluid velocity ${\bf u}$ in terms of the Clebsch potentials.}

  In Section 3.3 we transform the canonical 
Poisson bracket obtained from the Clebsch variable approach to a 
non-canonical Poisson bracket written in terms of Eulerian 
physical variables (see e.g. \cite{Morrison80,Morrison82a}, 
\cite{Morrison82},  
and \cite{Holm83a, Holm83b}).  
We obtain the non-canonical Poisson brackets for MHD using the variables
$({\bf M}, {\bf B}, \rho,\sigma)$ where ${\bf M}=\rho {\bf u}$ is the 
MHD momentum flux, $\sigma=\rho S$  and ${\bf B}$ is the magnetic induction. 
We also use the non-canonical variables $({\bf M}, {\bf A}, \rho,\sigma)$
where ${\bf A}$ is the magnetic vector potential in which the gauge is chosen 
so that the 1-form $\alpha={\bf A}{\bf\cdot} d{\bf x}$ is an invariant 
advected with the flow. 

\subsection{Clebsch Variables and the Momentum Map}
The Clebsch variational principle using the inverse map (i.e. 
Lagrangian map) involves the variational principle $\delta{\cal A}=0$ 
where
\beqn
{\cal A}=\int \ell [{\bf u}]\ d^3x\ dt+\int\ \boldsymbol{\pi}{\bf\cdot} 
\left({\bf l}_t+{\bf u}{\bf\cdot}\nabla{\bf l}\right)\ d^3x\ dt. 
\label{eq:mom1}
\eeqn
In (\ref{eq:mom1}) $\boldsymbol{\pi}$ is a Lagrange multiplier, 
ensuring that the Lagrange label $\bf{l}$ is advected with the fluid. 
In the present section we consider only the generic form of the 
Clebsch variational principle. More specific versions of the variational 
principle for MHD are discussed in later sections. 

The stationary point conditions for the variational functional (\ref{eq:mom1}) 
are:
\begin{align}
\frac{\delta{\cal A}} {\delta{\bf u}}=&\frac{\delta\ell}{\delta{\bf u}}
+\left(\nabla {\bf l}\right)^T{\bf\cdot}\boldsymbol{\pi}=0, \label{eq:mom2}\\
\frac{\delta{\cal A}}{\delta\boldsymbol{\pi}}
=&{\bf l}_t+{\bf u}{\bf\cdot}\nabla{\bf l}=0, \label{eq:mom3}\\
\frac{\delta{\cal A}}{\delta{\bf l}}=&-\left[\deriv{\boldsymbol{\pi}}{t}
+\nabla{\bf\cdot}({\bf u}\boldsymbol{\pi})\right]=0. \label{eq:mom4}
\end{align}
In ideal MHD 
\beqn
\ell=\frac{1}{2}\rho |{\bf u}|^2-\varepsilon (\rho, S)-\frac{B^2}{2\mu_0}\quad 
\hbox{and}\quad \frac{\delta \ell}{\delta{\bf u}}
=\rho {\bf u}\equiv {\bf m}, \label{eq:mom5}
\eeqn
where ${\bf m}$ is the fluid momentum density or mass flux (see Section 3.2 
for more details).

\subsubsection{Clebsch Momentum Map}
 The {\em momentum map} ${\bf J}: T^*Q\to \mathfrak{g}^*$ 
from the cotangent bundle $T^*Q$  of the configuration manifold $Q$ to the 
dual $\mathfrak{g}^*$ of the Lie algebra $\mathfrak{g}$ of Lie group 
$G$ that acts on $Q$ defines a momentum map by the formula:
\beqn
{\bf J}(\nu_q){\bf\cdot}\boldsymbol{\xi}
=\langle\nu_q,\boldsymbol{\xi}_Q(q)\rangle, \label{eq:mom6}
\eeqn
where $\nu_q\in T^*Q$ and $\boldsymbol{\xi}\in \mathfrak{g}$. Here
$\boldsymbol{\xi}_Q$ is the infinitesimal generator of the Lie algebra element 
$\boldsymbol{\xi}$ action of $\boldsymbol{\xi}$ on $Q$ 
and $\langle\nu_q,\boldsymbol{\xi}_Q(q)\rangle$ is the pairing of 
an element of $T^*Q$ with an element of $TQ$.

For the case (\ref{eq:mom1})-(\ref{eq:mom4}) 
the elements of $Q$ are the fluid labels ${\bf l}$ and the elements of 
$T^*Q$ are the conjugate pairs $({\bf l},\boldsymbol{\pi})$ of labels 
${\bf l}$ and their conjugate momenta $\boldsymbol{\pi}$. 

\begin{proposition}
The Clebsch relation (\ref{eq:mom2}) defines a right action ${\rm Diff}(\Omega)$ of 
diffeomorphisms on the domain $\Omega$ on the back to labels map ${\bf l}$.
\end{proposition}

\begin{proof}
Equation (\ref{eq:mom2}) defines a map 
$J_\Omega: T^*Q\to \mathfrak{X}
^*(\Omega)$ from the cotangent bundle $T^*Q$ 
to the dual $\mathfrak{X}^*(\Omega)$ of vector fields 
on $\Omega$:
\beqn
J_{\Omega}:\quad {\bf m}{\bf\cdot}d {\bf x}
=-\left(\left(\nabla{\bf l}\right)^T{\bf\cdot}\boldsymbol{\pi}\right)
{\bf\cdot}d{\bf x}=-\boldsymbol{\pi}{\bf\cdot}d{\bf l}\quad \hbox{where}\quad {\bf m}=\frac{\delta\ell}{\delta{\bf u}}. 
 \label{eq:mom7}
\eeqn
Thus, ${\bf J}_{\Omega}$ maps the conjugate pairs 
$({\bf l},\boldsymbol{\pi})$ to the space of 1-form densities ${\bf m}\in 
\mathfrak{X}^*(\Omega)$. 
The map $J_\Omega$ may be associated with 
the smooth, invertible maps or diffeomorphisms $\boldsymbol{\eta}$ of the 
back-to-labels maps ${\bf l}$ by composition of functions 
${\rm Diff}(\Omega):\ {\bf l}{\bf\cdot}\boldsymbol{\eta}
={\bf l}\circ\boldsymbol{\eta}$. The effect of the map 
on the infinitesimal generators $X_\Omega({\bf l})$ is formally defined as:
\beqn
X_\Omega({\bf l}):=\frac{d}{ds}({\bf l}\circ\boldsymbol{\eta})(s)|_{s=0}
=T{\bf l}\circ X, 
\label{eq:mom8}
\eeqn
where the differentiation with respect to $s$ is about the 
identity transformation at $s=0$. 
Equation (\ref{eq:mom8}) in component form is:
\begin{align}
X_\Omega({\bf l})^i\partial_{l^i}
=&\frac{d}{ds}{\bf l}[\boldsymbol{\eta}(s)]^i\partial_{l^i} 
=\deriv{l^i}{\eta^k} \frac{d\eta^k}{ds}\partial_{l^i}
=\frac{d\eta^k}{ds}\left(\deriv{l^i}{\eta^k}\derv{l^i}\right)
=\frac{d\eta^k}{ds}\derv{\eta^k}\nonumber\\
=&\frac{d\eta^k}{ds}\deriv{x^\mu}{\eta^k}\derv{x^\mu}=\frac{dx^\mu}{ds}\derv{x^\mu}
\equiv \left(T{\bf l}\circ X\right)^i\partial_{x^i}. 
\label{eq:mom9}
\end{align}
From (\ref{eq:mom9}) it follows by the chain rule for differentiation that:
\beqn
X=\frac{dx^\mu}{ds}\derv{x^\mu}=\frac{dl^i}{ds}\derv{l^i}
=\frac{d\eta^k}{ds}\derv{\eta^k}, 
\label{eq:mom10}
\eeqn
are equivalent forms for $X$.

The pairing between the map ${\bf J}_\Omega$ and the vector field 
$X\in \mathfrak{X}(\Omega)$ gives:
\begin{align}
\langle {\bf J}_\Omega\left({\bf l},\boldsymbol{\pi}\right),X\rangle=&
-\langle\boldsymbol{\pi}{\bf\cdot}d{\bf l},X\rangle
=-\langle\pi_k dl^k,X^j\partial_{x^j}\rangle=-\langle\pi_k l^k_{,s} dx^s, 
X^j\partial_{x^j}\rangle\nonumber\\
=&-\int \pi_kl^k_{,s}X^s\ d^3x=-\int\pi_k\left(X^sl^k_{,s}\right)\ d^3x
\nonumber\\ 
=&-\int\pi_k X\left(l^k\right)\ d^3x=-\int\pi_k\left(T{\bf l}\circ X\right)^k 
\ d^3x\nonumber\\
=&-\int\left({\bf l},\boldsymbol{\pi}\right)_k \left(T{\bf l}\circ X\right)^k 
\ d^3x
=-\langle\left({\bf l},\boldsymbol{\pi}\right),
X_\Omega\left({\bf l}\right)\rangle, \label{eq:mom11}
\end{align}
Thus,
\beqn
\langle {\bf J}_\Omega\left({\bf l},\boldsymbol{\pi}\right),X\rangle=
-\langle\left({\bf l},\boldsymbol{\pi}\right),
X_\Omega\left({\bf l}\right)\rangle, \label{eq:mom12}
\eeqn
which is equivalent to the defining relation (\ref{eq:mom6}) (the sign 
in (\ref{eq:mom6}) can be negative in the definition) for a momentum map.
\end{proof}    

\subsection{Clebsch variables and Hamilton's Equations}

Consider the MHD action (modified by constraints):
\beqn
J=\int\ d^3x\ dt  L,  \label{eq:Clebsch1}
\eeqn
where
\beqnar
L=&&\left\{\frac{1}{2}\rho u^2-\epsilon(\rho, S)-\frac{B^2}{2\mu_0}\right\}
+\phi\left(\deriv{\rho}{t}+\nabla{\bf\cdot}(\rho {\bf u})\right)\nonumber\\
&&+\beta\left(\deriv{S}{t}+{\bf u}{\bf\cdot}\nabla S\right)
+\lambda\left(\deriv{\mu}{t}+{\bf u\cdot}\nabla\mu\right) \nonumber\\
&&+\boldsymbol{\Gamma}{\bf\cdot}\left(\deriv{\bf B}{t}-\nabla\times({\bf u}\times{\bf B})
+{\bf u}(\nabla{\bf\cdot B})\right). \label{eq:Clebsch2}
\eeqnar
The Lagrangian in curly brackets equals the kinetic minus
the potential energy (internal thermodynamic energy plus magnetic energy).
The Lagrange multipliers $\phi$, $\beta$, $\lambda$, 
and $\boldsymbol{\Gamma}$ ensure that the 
mass, entropy, Lin constraint, Faraday equations are satisfied. We do not 
enforce $\nabla{\bf\cdot}{\bf B}= 0$, since we are interested in the 
effect of $\nabla{\bf\cdot}{\bf B}\neq 0$ (which is useful for numerical 
MHD where $\nabla{\bf\cdot}{\bf B}\neq 0$). {\bf It is straightforward 
to impose $\nabla{\bf\cdot}{\bf B}\neq 0$ if desired, 
and non-canonical Poisson brackets exist for this case as well (see 
(\ref{eq:H9}) in this paper, and \cite{Morrison82a}). Noncanonical Poisson 
brackets exist for this case as well as for the case 
$\nabla{\bf\cdot}{\bf B}=0$ 
(e.g. \cite{Morrison82a}, \cite{Chandre13}).}    

 Stationary point conditions for the action are $\delta J=0$.
 $\delta J/\delta {\bf u}=0$ gives the Clebsch representation
for ${\bf u}$:
\beqn
{\bf u}=\nabla\phi-\frac{\beta}{\rho}\nabla S-\frac{\lambda}{\rho}\nabla\mu+
{\bf u}_M\label{eq:Clebsch3}
\eeqn
where
\beqn
 {\bf u}_M=-\frac{(\nabla\times\boldsymbol{\Gamma})\times{\bf B}}{\rho}
-\boldsymbol{\Gamma}\frac{\nabla{\bf\cdot B}}{\rho}, \label{eq:Clebsch4}
\eeqn
is magnetic contribution to ${\bf u}$.
 Setting $\delta J/\delta\phi$, $\delta J/\delta\beta$,
$\delta J/\delta \lambda$, $\delta J/\delta\boldsymbol{\Gamma}$ consecutively 
equal to zero gives the mass, entropy advection, Lin constraint, 
and Faraday (magnetic flux conservation) constraint
equations:
\beqnar
&&\rho_t+\nabla{\bf\cdot}(\rho {\bf u})=0,\nonumber\\
&&S_t+{\bf u}{\bf\cdot}\nabla S=0,\nonumber\\
&&\mu_t+{\bf u\cdot}\nabla\mu=0, \nonumber\\
&&{\bf B}_t-\nabla\times({\bf u}\times{\bf B})+{\bf u}(\nabla{\bf\cdot B})=0.
\label{eq:Clebsch5}
\eeqnar

 Setting $\delta J/\delta\rho$, $\delta J/\delta S$, $\delta J/\delta\mu$,
$\delta J/\delta {\bf B}$ equal to zero gives evolution equations 
for the Clebsch potentials $\phi$, $\beta$, $\lambda$ and $\boldsymbol{\Gamma}$
as:
\beqnar
&&-\left(\deriv{\phi}{t}+{\bf u}{\bf\cdot}\nabla\phi\right)
+\frac{1}{2} u^2-h=0, \label{eq:Clebsch6}\\
&&\deriv{\beta}{t}+\nabla{\bf\cdot}(\beta {\bf u})+\rho T=0, 
\label{eq:Clebsch7}\\
&&\deriv{\lambda}{t}+\nabla{\bf\cdot}(\lambda {\bf u})=0, 
\label{eq:Clebsch8}\\
&&\deriv{\boldsymbol{\Gamma}}{t}
-{\bf u}\times (\nabla\times\boldsymbol{\Gamma})
+\nabla(\boldsymbol{\Gamma}{\bf\cdot u})+\frac{\bf B}{\mu_0}=0. 
\label{eq:Clebsch9}
\eeqnar
Equation (\ref{eq:Clebsch6}) is related to Bernoulli's equation for potential 
flow.The $\nabla(\boldsymbol{\Gamma}{\bf\cdot u})$ term 
in (\ref{eq:Clebsch9}) is associated with
$\nabla{\bf\cdot}{\bf B}\neq 0$.
 Taking the curl of (\ref{eq:Clebsch9})  gives:
\beqn
\deriv{\tilde{\boldsymbol{\Gamma}}}{t}
-\nabla\times({\bf u}\times {\tilde{\boldsymbol{\Gamma}}})
=-\frac{\nabla\times{\bf B}}{\mu_0} \quad\hbox{where}
\quad \tilde{\boldsymbol{\Gamma}}=\nabla\times\boldsymbol{\Gamma}.  \label{eq:Clebsch10}
\eeqn

 Equations (\ref{eq:Clebsch6})-(\ref{eq:Clebsch10}) can be written in the form:
\beqnar
&&\frac{d\phi}{dt}=\frac{1}{2}u^2-h,\quad 
\frac{d}{dt}\left(\frac{\beta}{\rho}\right)=-T, \nonumber\\
&&\frac{d}{dt}\left(\lambda d^3x\right)=0\quad \hbox{or}
\quad \frac{d}{dt}\left(\frac{\lambda}{\rho}\right)=0, \nonumber\\
&&\frac{d}{dt}(\boldsymbol{\Gamma}{\bf\cdot}d{\bf x})
=-\frac{{\bf B}{\bf\cdot}d{\bf x}}{\mu_0},\quad
\frac{d}{dt}(\tilde{\boldsymbol{\Gamma}}{\bf\cdot}d{\bf S})
=-{\bf J}{\bf\cdot}d{\bf S}. \label{eq:Clebsch11}
\eeqnar
where $d/dt=\partial/\partial t+{\bf u}{\bf\cdot}\nabla$,
is the Lagrangian time
derivative following the flow and ${\bf J}=\nabla\times{\bf B}/\mu_0$ 
is the current. $\boldsymbol{\Gamma}{\bf\cdot}d{\bf x}$ is a 1-form and 
$\tilde{\boldsymbol\Gamma}{\bf\cdot}d{\bf S}$ is a 2-form 
and $d{\bf S}$ is an area element.  

 Introduce the Hamiltonian functional:
\beqn
{\cal H}=\int H d^3x\quad\hbox{where}\quad H=\frac{1}{2}\rho u^2+\epsilon(\rho,S)+\frac{B^2}{2\mu_0}. \label{eq:H1}
\eeqn
 Substitute the Clebsch
expansion (\ref{eq:Clebsch3})-(\ref{eq:Clebsch4}) for ${\bf u}$
in (\ref{eq:H1}). Evaluating the  variational derivatives
of ${\cal H}$ gives Hamilton's equations:
\begin{align}
&\deriv{\rho}{t}=\frac{\delta{\cal H}}{\delta \phi},
\quad \deriv{\phi}{t}=-\frac{\delta {\cal H}}{\delta\rho}, \quad
\deriv{S}{t}=\frac{\delta{\cal H}}{\delta\beta},\quad
\deriv{\beta}{t}=-\frac{\delta{\cal H}}{\delta S}, \nonumber\\
&\deriv{\mu}{t}=\frac{\delta{\cal H}}{\delta{\lambda}}, \quad
\deriv{\lambda}{t}=-\frac{\delta{\cal H}}{\delta\mu}, \quad
\deriv{\bf B}{t}=\frac{\delta{\cal H}}{\delta\boldsymbol{\Gamma}}, \quad
\deriv{\boldsymbol{\Gamma}}{t}=-\frac{\delta{\cal H}}{\delta{\bf B}}. 
\label{eq:H3} 
\end{align}
Here  $\{\rho,\phi\}$, $\{S,\beta\}$,
$\{\mu,\lambda\}$, $\{{\bf B},\boldsymbol{\Gamma}\}$
are canonically conjugate variables.

The canonical Poisson bracket is:
\begin{align}
\{F,G\}=&\int d^3x\ \biggl(\frac{\delta F}{\delta\rho}
\frac{\delta G}{\delta \phi}-\frac{\delta F}{\delta \phi}
\frac{\delta G}{\delta\rho}
+\frac{\delta F}{\delta{\bf B}}
{\bf\cdot}\frac{\delta G}{\delta \boldsymbol{\Gamma}}
-\frac{\delta F}{\delta \boldsymbol{\Gamma}}
{\bf\cdot}\frac{\delta G}{\delta{\bf B}}\nonumber\\
&+\frac{\delta F}{\delta S}
\frac{\delta G}{\delta\beta}-\frac{\delta F}{\delta\beta}
\frac{\delta G}{\delta S}
+\frac{\delta F}{\delta\mu}
\frac{\delta G}{\delta\lambda}-\frac{\delta F}{\delta\lambda}
\frac{\delta G}{\delta\mu}\biggr).
 \label{eq:H4}
\end{align}
In terms of the Poisson bracket (\ref{eq:H4}) the time evolution of a functional $F$ of the canonical or physical variables is given by $F_t=\{F,H\}$ 
where $H$ is the Hamiltonian of the system. 
The canonical Poisson bracket (\ref{eq:H4})
satisfies the linearity, skew symmetry and Jacobi identity necessary
for a Hamiltonian system (i.e. the Poisson bracket defines a Lie algebra).
\begin{remark}
\ \cite{Cotter07} derive the Clebsch variable equations analogous 
to (\ref{eq:Clebsch3})-(\ref{eq:Clebsch10}), by using an advected or Lie dragging 
formalism, in which the advected quantity $a$ satisfies the Lie dragging equation:
\beqn
\left(\derv{t}+{\cal L}_{\bf u}\right) a=0, \label{eq:liedrag1}
\eeqn
where ${\cal L}_{\bf u}$ is the Lie derivative of $a$ with respect to the vector 
field ${\bf u}$. Examples of Lie dragged quantities  are $\omega^0=S$ (a scalar  
or $0$-form), $\omega^2={\bf B}{\bf\cdot}d{\bf S}$ is the Faraday 2-form, 
and $\omega^3=\rho d^3x$ is the mass 3-form (see also \cite{Webb14a} for further 
description of Lie dragged invariants in MHD). They also introduce fluid labels ${\bf l}^A$ 
which are advected with the flow, which are useful in specifying the initial state 
of the fluid, and their canonically conjugate momenta ${\bf\pi}_A$ which are the Lagrange multipliers
for ${\bf l}^A$ in the variational principle. They show how the Clebsch 
variational equations, the momentum map (Clebsch variational 
equation obtained by varying ${\bf u}$)  
can be combined to yield the Euler-Poincar\'e or Eulerian momentum equation for the system.  
\end{remark}
\subsection{Non-Canonical Poisson Brackets}
 \cite{Morrison80,Morrison82a} introduced non-canonical 
Poisson brackets for MHD. {\bf \cite{Morrison80} gave the 
non-canonical Poisson bracket for MHD for the case 
$\nabla{\bf\cdot}{\bf B}=0$. 
\cite{Morrison82a} gave the form of the Poisson bracket for 
$\nabla{\bf\cdot}{\bf B}\neq 0$.
A detailed discussion of the non-canonical Poisson bracket and the Jacobi 
identity is given by \cite{Morrison82}. \cite{Holm83a,Holm83b} point out
that their Poisson bracket has the form expected for a semi-direct product 
Lie algebra, for which the Jacobi identity is automatically satisfied. 
\cite{Chandre13} use Dirac's theory of constraints to derive properties 
of the Poisson bracket for the case $\nabla{\bf\cdot}{\bf B}=0$.}

Introduce the new variables:
\beqn
{\bf M}=\rho {\bf u}, \quad \sigma=\rho S, \label{eq:H5}
\eeqn 
noting that
\beqn
{\bf M}=\rho {\bf u}=\rho\nabla\phi-\beta\nabla S-\lambda\nabla\mu
+{\bf B}{\bf\cdot}(\nabla\boldsymbol{\Gamma})^T
-{\bf B}{\bf\cdot}\nabla\boldsymbol{\Gamma}
-\boldsymbol{\Gamma}(\nabla{\bf\cdot}{\bf B}), \label{eq:H7}
\eeqn
and transforming the canonical Poisson bracket (\ref{eq:H4}) from 
the old variables $(\rho,\phi,S,\beta,{\bf B},\boldsymbol{\Gamma})$ 
to the 
new variables $(\rho,\sigma,{\bf B},{\bf M})$  
we obtain the \cite{Morrison82a} non-canonical  
Poisson bracket:
\begin{align}
\left\{F,G\right\}=&-\int\ d^3x \biggl\{ \rho
\left[\frac{\delta F}{\delta{\bf M}}
{\bf\cdot}\nabla\left(\frac{\delta G}{\delta\rho}\right)
-\frac{\delta G}{\delta{\bf M}}
{\bf\cdot}\nabla\left(\frac{\delta F}{\delta\rho}\right)\right]\nonumber\\
&+\sigma\left[\frac{\delta F}{\delta{\bf M}}
{\bf\cdot}\nabla\left(\frac{\delta G}{\delta\sigma}\right)
-\frac{\delta G}{\delta{\bf M}}
{\bf\cdot}\nabla\left(\frac{\delta F}{\delta\sigma}\right)\right]\nonumber\\
&+{\bf M}{\bf\cdot}\left[\left(\frac{\delta F}{\delta{\bf M}}
{\bf\cdot}\nabla\right)\frac{\delta G}{\delta{\bf M}}
-\left(\frac{\delta G}{\delta{\bf M}}
{\bf\cdot}\nabla\right)\frac{\delta F}{\delta{\bf M}}\right]\nonumber\\
&+{\bf B}{\bf\cdot}\left[\frac{\delta F}{\delta{\bf M}}
{\bf\cdot}\nabla\left(\frac{\delta G}{\delta{\bf B}}\right)
-\frac{\delta G}{\delta{\bf M}}
{\bf\cdot}\nabla\left(\frac{\delta F}{\delta{\bf B}}\right)\right]\nonumber\\
&+{\bf B}{\bf\cdot}\left[
\left(\nabla\frac{\delta F}{\delta {\bf M}}\right){\bf\cdot}
\frac{\delta G}{\delta{\bf B}}
-\left(\nabla\frac{\delta G}{\delta{\bf M}}\right){\bf\cdot}
\frac{\delta F}{\delta{\bf B}}\right]
\biggr\}. \label{eq:H9}
\end{align}
The bracket (\ref{eq:H9}) has the Lie-Poisson form and 
satisfies the Jacobi identity for all functionals 
$F$ and $G$ of the physical variables, and in general applies both for 
$\nabla{\bf\cdot}{\bf B}\neq 0$  and $\nabla{\bf\cdot}{\bf B}=0$.

\subsubsection{Advected ${\bf A}$ Formulation}
Consider the MHD variational principle using the magnetic vector 
potential ${\bf A}$ instead of using ${\bf B}$ (e.g. 
\cite{Holm83a, Holm83b}). 
The condition that the magnetic 
flux ${\bf B}{\bf\cdot}d{\bf S}$ is Lie dragged with the flow (i.e. Faraday's 
equation) as a constraint equation, is satisfied if 
 the magnetic 
vector potential 1-form $\alpha={\bf A}{\bf\cdot}d{\bf x}$ is Lie dragged by 
the flow, where 
${\bf B}=\nabla\times{\bf A}$. 
The condition
that the one-form $\alpha={\bf A}{\bf\cdot}d{\bf x}$ is Lie dragged 
with the flow implies: 
\beqn
\deriv{\bf A}{t}-{\bf u}\times(\nabla\times{\bf A})+\nabla({\bf u\cdot A})=0 
\label{eq:can1}
\eeqn
{\bf (see \cite{Gordin87}, \cite{Padhye96a,Padhye96b}, \cite{Webb14a})}. 
The condition that the magnetic flux 
$\beta=d\alpha={\bf B}{\bf\cdot}d{\bf S}$ is Lie dragged with the flow 
implies Faraday's equation:
\beqn
\deriv{\bf B}{t}-\nabla\times({\bf u}\times{\bf B})
+{\bf u}\nabla{\bf\cdot}{\bf B}=0 \label{eq:can1a}
\eeqn
for the magnetic induction ${\bf B}$. Note that the curl of  (\ref{eq:can1}) 
with ${\bf B}=\nabla\times {\bf A}$ gives Faraday's equation (\ref{eq:can1a}) 
where $\nabla{\bf\cdot}{\bf B}=0$. 

 We use the variational 
principle $\delta {\cal A}=0$ where the action ${\cal A}$ is given by:
\begin{align}
{\cal A}=&\int_Vd^3x\int dt\biggl\{\left[\frac{1}{2}\rho |{\bf u}|^2 
-\varepsilon(\rho,S)-\frac{|\nabla\times{\bf A}|^2}{2\mu_0}\right]\nonumber\\
&+\phi\left(\deriv{\rho}{t}+\nabla{\bf\cdot}(\rho{\bf u})\right) 
+\beta \left(\deriv{S}{t}+{\bf u}{\bf\cdot}\nabla S\right)
+\lambda \left(\deriv{\mu}{t}+{\bf u}{\bf\cdot}\nabla \mu\right)\nonumber\\
&+\boldsymbol{\gamma}{\bf\cdot}
\left[\deriv{\bf A}{t}-{\bf u}\times(\nabla\times{\bf A})
+\nabla({\bf u\cdot A})\right]\biggr\}. \label{eq:can2}
\end{align}

By setting the variational derivative $\delta{\cal A}/\delta {\bf u}=0$ gives 
the Clebsch variable expansion:
\beqn
{\bf u}=\nabla\phi-\frac{\beta}{\rho}\nabla S
-\frac{\lambda}{\rho} \nabla\mu-\frac{\boldsymbol{\gamma}\times
(\nabla\times{\bf A})}{\rho}+\frac{\nabla{\bf\cdot}\boldsymbol{\gamma}}{\rho}
{\bf A}, \label{eq:can3}
\eeqn
for the fluid velocity ${\bf u}$. 

Setting the variational derivatives $\delta{\cal A}/\delta\phi$, 
$\delta{\cal A}/\delta\beta$, $\delta{\cal A}/\delta\lambda$, 
$\delta{\cal A}/\delta\boldsymbol{\gamma}$ 
equal to zero gives the constraint equations:
\begin{align}
&\deriv{\rho}{t}+\nabla{\bf\cdot}(\rho {\bf u})=0,\quad \deriv{S}{t}
+{\bf u}{\bf\cdot}\nabla S=0, \nonumber\\
&\deriv{\mu}{t}+{\bf u}{\bf\cdot}\nabla\mu=0,\nonumber\\
&\deriv{\bf A}{t}-{\bf u}\times(\nabla\times{\bf A})+\nabla({\bf u\cdot A})=0.
\label{eq:can4}
\end{align}
Similarly setting $\delta{\cal A}/\delta\rho$, $\delta{\cal A}/\delta S$, 
$\delta{\cal A}/\delta \mu$ 
and $\delta{\cal A}/\delta {\bf A}$ equal to zero gives the equations:
\begin{align}
&\deriv{\phi}{t}+{\bf u}{\bf\cdot}\nabla\phi+h-\frac{1}{2}|{\bf u}|^2=0, 
\nonumber\\
&\deriv{\beta}{t}+\nabla{\bf\cdot}(\beta {\bf u})+\rho T=0, 
\quad \deriv{\lambda}{t}+\nabla{\bf\cdot}(\lambda {\bf u})=0, 
\nonumber\\
&\deriv{\boldsymbol{\gamma}}{t}
-\nabla\times({\bf u}\times\boldsymbol{\gamma})
+{\bf u}(\nabla{\bf\cdot}\boldsymbol{\gamma}) 
+\frac{\nabla\times{\bf B}}{\mu}=0. \label{eq:can5}
\end{align}
The Euler-Lagrange equations (\ref{eq:can3})-(\ref{eq:can5}) together imply 
Hamilton's equations:
\begin{align}
&\deriv{\rho}{t}=\frac{\delta{\cal H}}{\delta \phi},
\quad \deriv{\phi}{t}=-\frac{\delta {\cal H}}{\delta\rho}, \quad
\deriv{S}{t}=\frac{\delta{\cal H}}{\delta\beta},\quad
\deriv{\beta}{t}=-\frac{\delta{\cal H}}{\delta S}, \nonumber\\
&\deriv{\mu}{t}=\frac{\delta{\cal H}}{\delta\lambda},\quad
\deriv{\lambda}{t}=-\frac{\delta{\cal H}}{\delta \mu},\quad 
\deriv{\bf A}{t}=\frac{\delta{\cal H}}{\delta\boldsymbol{\gamma}}, \quad
\deriv{\boldsymbol{\gamma}}{t}=-\frac{\delta{\cal H}}{\delta{\bf A}}. 
\label{eq:can6} 
\end{align}
Here  $\{\rho,\phi\}$, $\{S,\beta\}$,
and $\{{\bf A},\boldsymbol{\gamma}\}$
are canonically conjugate variables. The Hamiltonian functional 
${\cal H}$ is given by (\ref{eq:H1}), and ${\bf u}$ is given by the Clebsch 
expansion (\ref{eq:can3}). 
The canonical Poisson bracket is:
\begin{align}
\{F,G\}=&\int d^3x\ \biggl(\frac{\delta F}{\delta\rho}
\frac{\delta G}{\delta \phi}-\frac{\delta F}{\delta \phi}
\frac{\delta G}{\delta\rho}
+\frac{\delta F}{\delta{\bf A}}
{\bf\cdot}\frac{\delta G}{\delta\boldsymbol{\gamma}}
-\frac{\delta F}{\delta\boldsymbol{\gamma}}
{\bf\cdot}\frac{\delta G}{\delta{\bf A}}\nonumber\\
&+\frac{\delta F}{\delta S}
\frac{\delta G}{\delta\beta}-\frac{\delta F}{\delta\beta}
\frac{\delta G}{\delta S} 
+\frac{\delta F}{\delta \mu}
\frac{\delta G}{\delta\lambda}-\frac{\delta F}{\delta\lambda}
\frac{\delta G}{\delta \mu}
\biggr). \label{eq:can7}
\end{align}

The transformations of the  
variational derivatives from canonical Clebsch variables 
$(\rho,\phi,S,\beta,{\bf A},\boldsymbol{\gamma})$ in terms of the 
non-canonical new variables $(\rho,\sigma,{\bf A},{\bf M})$ are:
\begin{align}
&\frac{\delta F}{\delta\rho}
=\frac{\delta F}{\delta\rho}+S\frac{\delta F}{\delta\sigma}
+\frac{\delta F}{\delta{\bf M}}{\bf\cdot}\nabla\phi,\quad
\frac{\delta F}{\delta\phi}=
-\nabla{\bf\cdot}\left(\rho\frac{\delta F}{\delta{\bf M}}\right),
\nonumber\\
&\frac{\delta F}{\delta S}=\rho\frac{\delta F}{\delta\sigma}
+\nabla{\bf\cdot}\left(\beta\frac{\delta F}{\delta {\bf M}}\right),
\quad
\frac{\delta F}{\delta \beta}
=-\frac{\delta F}{\delta {\bf M}}{\bf\cdot}\nabla S, \nonumber\\
&\frac{\delta F}{\delta{\bf A}}
=\frac{\delta F}{\delta{\bf A}}+\nabla{\bf\cdot}\boldsymbol{\gamma} 
\frac{\delta F}{\delta{\bf M}}
-\nabla\times
\left(\frac{\delta F}{\delta{\bf M}}\times\boldsymbol{\gamma}\right), 
\nonumber\\
&\frac{\delta F}{\delta\boldsymbol{\gamma}}
=-{\bf B}\times\frac{\delta F}{\delta{\bf M}} 
-\nabla\left[{\bf A}{\bf\cdot}
\left(\frac{\delta F}{\delta{\bf M}}\right)\right],\nonumber\\
&\frac{\delta F}{\delta\mu}=\nabla{\bf\cdot}\left(\lambda\frac{\delta F}
{\delta{\bf M}}\right),
\quad \frac{\delta F}{\delta\lambda}=-\frac{\delta F}{\delta{\bf M}}{\bf\cdot}\nabla\mu.  
\label{eq:can8}
\end{align}

In terms of the non-canonical variables $({\bf M},{\bf A},\rho,\sigma)$ 
where $\sigma=\rho S$ we obtain the non-canonical Poisson bracket:
\begin{align}
\left\{F,G\right\}=&-\int\ d^3x\biggl\{ \left[F_{\bf M}{\bf\cdot}
\nabla (G_{\bf M})-G_{\bf M}{\bf\cdot}\nabla(F_{\bf M})\right]
{\bf\cdot}{\bf M}
\nonumber\\
&+\rho\left[F_{\bf M}{\bf\cdot}\nabla(G_{\rho})
-G_{\bf M}{\bf\cdot}\nabla(F_{\rho})\right]\nonumber\\
&+\sigma\left[F_{\bf M}{\bf\cdot}\nabla(G_{\sigma})
-G_{\bf M}{\bf\cdot}\nabla(F_{\sigma})\right]\nonumber\\
&+{\bf A}{\bf\cdot}\left[F_{\bf M}\nabla{\bf\cdot}(G_{\bf A})
-G_{\bf M}\nabla{\bf\cdot}(F_{\bf A})\right]\nonumber\\
&+\nabla\times{\bf A}{\bf\cdot}
\left[G_{\bf A}\times F_{\bf M}-F_{\bf A}\times G_{\bf M}\right]\biggr\}, 
\label{eq:can9}
\end{align}
where $F_{\bf M}\equiv \delta F/\delta {\bf M}$ and similarly for the 
other variational derivatives in (\ref{eq:can9}). The non-canonical 
bracket (\ref{eq:can9}) was obtained by \cite{Holm83a, Holm83b}. 
It is a skew symmetric bracket and satisfies the Jacobi identity. 
\cite{Holm83a, Holm83b} show that bracket (\ref{eq:can9}) 
 corresponds to 
a semi-direct product Lie algebra.

\section{Overview of Multi-symplectic Systems}

\index{Hamiltonian}%
In this section we discuss multi-symplectic systems of partial differential 
equations. In Section 4.1, we set out the general theory, for the case 
where the space-time metric is flat. We assume the independent variables 
$x^\alpha$ have a flat metric. For example in  MHD we take the independent
variables as $(t,x,y,z)$ where $(x,y,z)$ are Cartesian space coordinates 
and $t$ is the time. In Section 4.2 we indicate how the analysis is altered
if the  independent variables are generalized coordinates (e.g. 
for spherical or cylindicral symmetry the metric plays a role via 
the replacement of ordinary partial derivatives by co-variant derivatives). 
\cite{Bridges10} also describe the general case of multi-symplectic 
systems taking into account the geometry of the independent variables.
They use the algebra of exterior differential forms and the
variational bi-complex, in which the exterior differential $d=d_h+d_v$
where $d_h$ is the so-called horizontal exterior derivative and $d_v$ is the 
vertical exterior derivative.  
\subsection{Flat Cartesian metric}

Hamiltonian systems, with one evolution
variable $t$, can in general be written in the form:
\beqn
{\sf K}_{ij}(z) \frac{dz^j}{dt}=\nabla_{z^i} H(z), \label{eq:m1}
\eeqn
where
the invariant phase space volume element:
\beqn
\kappa=\frac{1}{2} {\sf K}_{ij}(z) dz^i\wedge dz^j, \label{eq:m2}
\eeqn
is a closed two-form, i.e. $d\kappa=0$. Here $d$ denotes 
the exterior derivative and $\wedge$ denotes the anti-symmetric wedge 
product used in the exterior Calculus. The condition 
that $\kappa$ be a closed 2-form, implies by the Poincar\'e Lemma, that 
$\kappa=d g$ where $g=L_j dz^j$ is a one-form (note that $d\kappa =ddg=0$
by antisymmetry of the wedge product). It turns out, that the condition
that $\kappa$ be a closed 2-form implies that ${\sf K}_{ij}=-{\sf K}_{ji}$ is a 
skew symmetric operator (see \cite{Zakharov97}, 
and \cite{Hydon05}).
 Taking the exterior derivative 
of the 2-form (\ref{eq:m2}) and setting the result equal to zero, 
we obtain the identity:
\beqn
{\sf K}_{ij,k}+{\sf K}_{jk,i}+{\sf K}_{ki,j}=0, \label{eq:m3} 
\eeqn
which in some cases is related to the Jacobi identity for the Poisson bracket. 
If the 
system (\ref{eq:m1}) has an even dimension, and if ${\sf K}_{ij}$ has non-zero 
determinant
, then (\ref{eq:m1}) can be written in the form:
\beqn
\frac{dz^i}{dt}= {\sf R}_{ij} \nabla_{z^j} H(z), \label{eq:m3a}
\eeqn
where ${\sf R}_{ij}$ is the inverse of the 
matrix ${\sf K}_{ij}$. Here ${\sf R}_{ij}=-{\sf R}_{ji}$
is a skew-symmetric matrix.  
The closure relation (\ref{eq:m3}) then are equivalent to the relations:
\beqn
{\sf R}_{im} \deriv{{\sf R}_{jk}}{z^m}+{\sf R}_{km}\deriv{{\sf R}_{ij}}{z^m}
+{\sf R}_{jm}\deriv{{\sf R}_{ki}}{z^m}=0,
\label{eq:m3b}
\eeqn
(see e.g. \cite{Zakharov97}). The Poisson bracket for 
the system in the finite dimensional case is given by 
\beqn
\{A,B\}=\sum {\sf R}_{ij}\deriv{A}{z^i}\deriv{B}{z^j}
\label{eq:m3c}
\eeqn
Using the Poisson bracket description (\ref{eq:m3c}) 
the Jacobi identity reduces to  (\ref{eq:m3b}). Casimir functions or 
more generally functionals have zero Poisson bracket with respect to 
any other functional of the variables describing the system. For 
finite dimensional systems Casimirs always occur for odd dimensional 
systems.\\  


A finite dimensional
\index{Hamiltonian}%
Hamiltonian system of dimension $2n$
with canonical variables   $z=(q^1,q^2,\ldots q^n,
p_1,p_2,\ldots p_n)^t$ can be written in the form (\ref{eq:m1}), 
where
\beqn
{\sf K}={\sf J}^t=\left(\begin{array}{cc}
0&-I_n\\
I_n& 0
\end{array}
\right). \label{eq:m4a}
\eeqn
Here the matrix ${\sf K}$ is the inverse of the
\index{symplectic}%
symplectic matrix ${\sf J}$
and $I_n$ is the unit $n\times n$ matrix. The invariant phase space element
form (\ref{eq:m2}) is:
\beqn
\kappa= dp_j\wedge dq^j=d(p_j dq^j). \label{eq:m4}
\eeqn
\index{Hamiltonian}%
\index{symplectic}%
Hamiltonian, multi-symplectic systems with $n$ independent 
variables $x^\alpha$ can be written in the form: 
\beqn
{\sf K}_{ij}^\alpha z^j_{,\alpha}=\nabla_{z^i} H(z), \label{eq:m5}
\eeqn
where $z^j_{\alpha}=\partial z^j/\partial x^\alpha$. The 
fundamental invariant 2-forms are:
\beqn
\kappa^\alpha=\frac{1}{2} {\sf K}_{ij}^\alpha dz^i\wedge dz^j, \quad \alpha=1(1)n, 
\label{eq:m6}
\eeqn
Invariance of the phase space element $D_t(dp_j\wedge dq^j)=0$ for 
the standard canonical
\index{Hamiltonian}%
Hamiltonian formulation with evolution variable $t$
is replaced by the
\index{symplectic}%
symplectic, or structural
\index{conservation laws}%
conservation law:
\beqn
\kappa^\alpha_{,\alpha}=0, \label{eq:m7}
\eeqn
which is referred to as the
\index{symplectic}%
\index{symplecticity}%
symplecticity conservation law.

The closure of the 2-forms $\kappa^\alpha$ implies that the exterior 
derivative of $\kappa^\alpha=0$. By the Poincar\'e Lemma 
$\kappa^\alpha$ is the exterior derivative of a 1-form, i.e.,
\beqn
\kappa^\alpha=d(L_j^\alpha dz^j)=d\omega^\alpha
\quad\hbox{where}\quad \omega^\alpha=L_j^\alpha dz^j. \label{eq:m8}
\eeqn
Note that $d\kappa^\alpha=dd\omega^\alpha=0$. Taking the exterior derivative 
of $\omega^\alpha$ in (\ref{eq:m8}) and using the anti-symmetry of the wedge
product we obtain:

\beqn
\kappa^\alpha=\frac{1}{2} 
\left( \deriv{L_k^\alpha}{z^j}-\deriv{L_j^\alpha}{z^k}\right) dz^j\wedge dz^k. \label{eq:m9}
\eeqn
From (\ref{eq:m6}) and (\ref{eq:m9}) we obtain:
\beqn
{\sf K}_{jk}^\alpha=\deriv{L_k^\alpha}{z^j}-\deriv{L_j^\alpha}{z^k}. \label{eq:m10}
\eeqn
Thus, the matrices ${\sf K}_{ij}^\alpha$ are skew-symmetric, 
i.e. ${\sf K}_{ij}^\alpha=-{\sf K}_{ji}^\alpha$. 

\begin{proposition}\label{prop:4.1}
The Legendre transformation for
\index{symplectic}%
multi-symplectic systems
is the identity
\beqn
\left( L_j^\alpha dz^j\right)_{,\alpha}
=d\left\{L_j^\alpha(z) z^j_{,\alpha}-H(z)\right\}
\equiv dL, \label{eq:m11}
\eeqn
where
\beqn
L=L_j^\alpha(z) z^j_{,\alpha}-H(z), \label{eq:m12}
\eeqn
is the
\index{Lagrangian!density}%
Lagrangian density and $H(z)$ is the multi-symplectic Hamiltonian.
\end{proposition}

\begin{proof}
The proof of (\ref{eq:m11}) proceeds by noting
\begin{align}
\left(L_j^\alpha dz^j\right)_{,\alpha}=&\deriv{L_j^\alpha}{z^i} z_{,\alpha}^i dz^j
 +L_j^\alpha (z) D_\alpha d z^j\nonumber\\
=&\deriv{L_j^\alpha}{z^i} z_{,\alpha}^i dz^j
 +L_j^\alpha (z)  d \left(z^j_{,\alpha}\right). \label{eq:m13}
\end{align}
Here we used the fact that the operators $d$ and $D_\alpha$ commute. 
(\ref{eq:m13}) can be further reduced to:
\beqn
\left(L_j^\alpha dz^j\right)_{,\alpha}=-{\sf K}_{ji}^\alpha z_{,\alpha}^i dz^j 
+d\left(
L_j^\alpha(z) z_{,\alpha}^j\right). \label{eq:m14}
\eeqn
The identity (\ref{eq:m11}) then follows by using the
\index{Hamiltonian}%
Hamiltonian
evolution equations (\ref{eq:m5}).
\end{proof}

 The
\index{symplectic}%
\index{symplecticity}%
symplecticity or structural
\index{conservation laws}%
conservation law (\ref{eq:m7}) now follows by taking the exterior
derivative of (\ref{eq:m11}) and using the results $ddL=0$ and
$dD_\alpha=D_\alpha d$, i.e.,
\beqn
D_\alpha \kappa^\alpha=D_\alpha[d(L^\alpha_j dz^j)]
=dD_\alpha(L^\alpha_j dz^j)=ddL=0, \label{eq:m15}
\eeqn
which is (\ref{eq:m7}). 
 Other conservation laws are obtained by
sectioning the forms in (\ref{eq:m11}) (i.e. we impose the requirement
that $z^j=z^j({\mathbf x})$, which is also referred to as the pull-back to
the base manifold). The pullback, applied to (\ref{eq:m11}) gives
\beqn
\left(L_j^\alpha dz^j\right)_{,\alpha} 
=\left(L_j^\alpha z^j_{,\beta} dx^\beta\right)_{,\alpha}
=\left(L_j^\alpha z^j_{,\beta}\right)_{,\alpha} dx^\beta 
=dL=\deriv{L}{x^\beta}\ dx^\beta. \label{eq:m16}
\eeqn
 Thus, (\ref{eq:m16})  gives the conservation law:
\beqn
D_\alpha\left(L_j^\alpha(z)z_{,\beta}^j-L\delta^\alpha_\beta\right)=0. 
\label{eq:m18}
\eeqn
This
\index{conservation laws}%
conservation law is in fact, the conservation law obtained due to
the invariance of the action $A=\int L dx$ under translations in $x^\beta$
which follows from Noether's first theorem (i.e. 
$x^{'\alpha}=x^\alpha+\epsilon \delta^\alpha_\beta$). 

A further set of 
$n(n-1)/2$ conservation laws is obtained from pull-back of the structural
conservation law (\ref{eq:m7}) to the base manifold, namely:
\beqn
D_\alpha\left( {\sf K}_{ij}^\alpha z_{,\beta}^i z^j_{,\gamma}\right)=0,\quad 
\beta<\gamma. \label{eq:m19}
\eeqn
The
\index{conservation laws}%
conservation laws (\ref{eq:m19}) can be obtained by cross-differentiation
of the conservation laws (\ref{eq:m15}), i.e. they are a consequence
of the equations:{\bf
\beqn
D_\gamma G_\beta-D_\beta G_\gamma=D_\alpha\left( {\sf K}_{ij}^\alpha 
z_{,\gamma}^i z^j_{,\beta}\right), \label{eq:m20}
\end{equation}
where
\begin{equation}
G_\beta=D_\alpha\left(L^\alpha_j z^j_{,\beta}-L\delta^\alpha_\beta\right).
\label{eq:m21}
\end{equation}
Here $G_\beta=0$ give the energy and momentum conservation equations 
for $\beta=0,1,2,3$ respectively.} 
A multi-symplectic version of Noether's theorem 
for the multi-symplectic system (\ref{eq:m5}) is described below 
(see also \cite{Hydon05}):
\begin{proposition}
If the action:
\begin{equation}
J=\int L\ d^3x dt \label{eq:msn1}
\end{equation}
is invariant to $O(\epsilon)$ under the infinitesimal Lie transformation:
\begin{equation}
z^{'s}=z^s +\epsilon V^{z^s},\quad x^{'\alpha}=x^\alpha+\epsilon V^{x^\alpha}, 
\quad (0\leq \alpha\leq 3,\quad 1\leq s\leq N), \label{eq:msn2}
\end{equation}
and under the divergence transformation:
\begin{equation}
L'=L+\epsilon D_\alpha \Lambda^\alpha+O(\epsilon^2), \label{eq:msn3}
\end{equation}
where $L$ has the multi-symplectic form (\ref{eq:m12}):
\begin{equation}
L=L_j^\alpha(z)z^j_{,\alpha}-H(z)\quad \hbox{and}\quad {\cal H}=
\int H(z) d^3x dt, \label{msn4}
\end{equation}
is the Hamiltonian functional, then the Euler Lagrange equations 
for the action:
\begin{equation}
E_{z^s}(L)= \deriv{L}{z^s}
-\derv{x^\alpha}\left(\deriv{L}{z^s_{,\alpha}}\right)
\equiv -\deriv{H}{z^s}+{\sf K}^\alpha_{sj} z^j_{,\alpha}=0, \label{eq:msn5}
\end{equation}
admit the conservation law 
\begin{equation}
D_\alpha\left\{ V^{x^\alpha}L+W^\alpha[{\bf z},\hat{V}^{\bf z}]
+\Lambda^\alpha\right\}=0. \label{eq:msn6}
\end{equation} 
In (\ref{eq:msn6})
\begin{equation}
 W^\alpha[{\bf z},\hat{V}^{\bf z}]
=\hat{V}^{z^s}\deriv{L}{z^s_{,\alpha}}
\equiv\hat{V}^{z^s} L^\alpha_s (z),  \label{eq:msn7}
\end{equation}
and
\begin{equation}
\hat{V}^{z^s}=V^{z^s}-V^{x^\alpha} z^s_{,\alpha}, \label{eq:msn8}
\end{equation}
is the canonical or characteristic Lie symmetry generator (i.e., the 
infinitesimal Lie symmetry transformation $z^{'s}=z^s+\epsilon \hat{V}^{z^s}$, 
$x^{'\alpha}=x^\alpha$ which is equivalent to Lie 
transformation (\ref{eq:msn2})).
Thus, the conservation law (\ref{eq:msn6}) reduces to:
\begin{equation}
D_\alpha \left\{V^{x^\alpha} 
\left[L^\mu_s({\bf z}) z^s_{,\mu}-H({\bf z})\right] 
+\hat{V}^{z^s} L^\alpha_s({\bf z})+\Lambda^\alpha\right\}=0. 
\label{eq:msn9}
\end{equation}
or aternatively:
\begin{equation}
D_\alpha\left\{V^{x^\alpha} L+\hat{V}^{z^s} L^\alpha_s({\bf z})
+\Lambda^\alpha\right\}=0. 
\label{eq:msn10}
\end{equation}
This is the multi-symplectic form of Noether's first theorem for the 
system (\ref{eq:m5}). 

The condition for the Lie symmetry (\ref{eq:msn2})-(\ref{eq:msn3})
to be a divergence symmetry of the action is:
\begin{equation}
\tilde{X}L +V^{x^\alpha}D_\alpha L 
+ D_\alpha \Lambda^\alpha=0, \label{eq:msn11}
\end{equation}
where
\begin{equation}
\tilde{X}=V^{x^\alpha}\derv{x^\alpha}+V^{z^s}\derv{z^s}
+V^{z^s_{,\alpha}}\derv{z^s_{,\alpha}}+\ldots, \label{eq:msn12}
\end{equation}
is the extended Lie symmetry operator. 
The extended Lie symmetry operator ${\tilde{X}}$ can be expressed in 
terms of the characteristic symmetry operator ${\hat{X}}$
by the formula
\begin{equation}
\tilde{X}=\hat{X}+V^{x^{\alpha}} D_\alpha,\quad\hbox{where}\quad 
\hat{X}=\hat{V}^{z^s}\derv{z^s}
+D_\alpha\left({\hat V}^{z^s}\right) 
\derv{z^s_{,\alpha}}+\ldots. \label{eq:msn13}
\end{equation}
The Lie invariance condition (\ref{eq:msn11}) written in terms of $\hat{X}$ is:
\begin{equation}
{\hat X}L+D_\alpha \left(V^{x^\alpha} L+\Lambda^\alpha\right)=0. 
\label{eq:msn14}
\end{equation}
\end{proposition}
\begin{example}
\ As an example of Noether's theorem, consider the invariance of the 
action $J$ 
under the Lie symmetry:
\begin{equation}
V^{x^\alpha}=\delta^\alpha_\beta,\quad V^{z^s}=0,\quad \Lambda^\alpha=0, 
 \label{eq:msn15}
\end{equation}
corresponding to translation invariance with respect to $x^\beta$. 
The canonical Lie symmetry generator $\hat{V}^{z^s}$ is given by:
\begin{equation}
\hat{V}^{z^s}=-z^s_{,\beta}, \label{eq:msn16}
\end{equation}
The Lie invariance condition (\ref{eq:msn14}) is satisfied for 
$\Lambda^\alpha=0$, i.e. the action is invariant under a variational symmetry
(one can show $\hat{X}L=-D_\beta L$). The conservation law (\ref{eq:msn9}) 
or (\ref{eq:msn10}) reduces to the symplectic conservation law  (\ref{eq:m18}).
Thus we have shown that the symplectic conservation law is due to 
invariance of the action under translations in $x^\beta$
\end{example}

\begin{proposition}\label{propforms}
The multi-symplectic system (\ref{eq:m5}) for the case where $x^\alpha=(t,x,y,z)$ is 
equivalent to the differential form system:
\begin{align}
\Omega=&d\omega^0\wedge dx\wedge dy\wedge dz-d\omega^1\wedge dx^0\wedge dy\wedge dz
+d\omega^2\wedge dx^0\wedge dx\wedge dz\nonumber\\
&-d\omega^3\wedge dx^0\wedge dx\wedge dy -dH\wedge dx^0\wedge dx\wedge dy\wedge dz=0, 
\label{eq:msn17}
\end{align}
in which the $z^\mu$ are the dependent variables and the $x^\alpha$ are the independent variables. 
The variational principle
\beqn
J=\int\Omega, \label{eq:msn18}
\eeqn
with $\delta J/\delta z^\mu=0$ gives the multi-symplectic 
system (\ref{eq:m5}). The form (\ref{eq:msn17}) is referred to as the Cartan-Poincar\'e form 
(e.g. \cite{Marsden99}).
\end{proposition}
\begin{proof}
\ Starting from (\ref{eq:m5}) we require:
\beqn
\left({\sf K}^\alpha_{\mu\nu} \deriv{z^\nu}{x^\alpha}-\frac{\delta H}{\delta z^\mu}\right) dz^\mu\wedge dx^0\wedge dx\wedge dy\wedge dz=0. \label{eq:msn19}
\eeqn
Equation (\ref{eq:msn19}) can be expanded to give (\ref{eq:msn17}). 
{\bf In the derivation of 
(\ref{eq:msn19}) the forms are sectioned, i.e. the variables $z^\mu$ are taken to be 
dependent on the independent variables $x^\alpha$, $\alpha=0,1,2,3$.}  

{\bf Next we consider the variation of $\delta J$ in (\ref{eq:msn18}), which may be reduced to the 
form:
\begin{equation}
\delta J=\int d(\delta z^\mu)\wedge \beta_\mu+\int \delta z^\mu d\beta_\mu, \label{eq:msn19a}
\end{equation}
where 
\beqn
\beta_\mu=(-1)^\alpha\left[{\sf K}^\alpha_{\mu\nu} dz^\nu\wedge 
\prod_{j\neq\alpha}dx^j\right] 
-\deriv{H}{z^\mu}\wedge  \prod_{j=0}^{n}dx^j, \label{eq:msn19b}
\eeqn
$\prod_{j=0}^{n}dx^j=dx^0\wedge dx^1\ldots\wedge dx^n$ and 
$\prod_{j\neq\alpha}dx^j$ is defined similarly, 
but does not include $dx^\alpha$ in the wedge product. 
$n$ is the number of space variables.
In the derivation of (\ref{eq:msn19a}) we used the representation:
\beqn
\Omega=(-1)^\alpha {\sf K}^\alpha_{\mu\nu} dz^\mu\wedge dz^\nu\wedge 
\prod_{j\neq\alpha}dx^j
-\deriv{H}{z^\mu} dz^\mu \wedge \prod_{j=0}^{n}dx^j\equiv dz^\mu\wedge\beta_\mu,
 \label{eq:msn19c}
\eeqn
where we choose $\mu<\nu$ in the first term in (\ref{eq:msn19c}), but there is no restriction 
on $\nu$ in (\ref{eq:msn19b}). The first term in (\ref{eq:msn19a}) is due to the variations 
of $\delta dz^\mu$, whereas the second term is due to the variations $\delta {\sf K}^\alpha_{\mu\nu}$ 
and $\delta(\partial H/\partial z^\mu)$. The second integral in (\ref{eq:msn19a}) vanishes by 
Stokes theorem:
\beqn
\int_Vd\beta_\mu=\int_{\partial V}\beta_\mu=0, \label{eq:msn19d}
\eeqn
where we assume that $\beta_\mu$ vanishes on the boundary $\partial V$. For independent variations 
of $\delta z^\mu$, the condition $\delta J=0$ gives the differential form system:
\beqn
\beta_\mu=0,\quad 1\leq\mu\leq N, \label{eq:msn19e}
\eeqn
where $N$ is the number of dependent variables $z^\mu$. 

Equations (\ref{eq:msn19e}) form the basis of the Cartan differential form representation of the 
differential equation system (\ref{eq:m5}). However, the ideal of forms representing 
the differential equation system (\ref{eq:m5}) in general needs to be 
enlarged to include the exterior derivatives 
of $d\beta_\mu$ which are not expressible as a linear combination of the $\beta_\mu$. 
The closed system of forms consisting of the $\beta_\mu$ plus the adjoined forms $d\beta_\nu$
can then be used to represent the differential equation system (\ref{eq:m5}) and the integrability 
conditions for the system (\ref{eq:m5}) (see e.g. \cite{Harrison71}). This 
completes our discussion of the variational principle (\ref{eq:msn18}).} 
\end{proof}

\subsection{Covariant Formulation}
In Section 4.1 we assumed that  
that the $L^\alpha_j$ and the ${\sf K}^\alpha_{ij}$  depended only on the field
variables ${\bf z}$. \cite{Bridges10} in a more general formulation 
using the variational bi-complex consider cases where the 
$L^\alpha_j$ and $H$ can also depend on the independent variables $x^k$. 
They show that the multi-symplectic system (\ref{eq:m5}) 
is a special case of the 
more general system:
\beqn
{\sf K}^\alpha_{ij}\deriv{z^j}{q^\alpha}-L^\alpha_{i;\alpha}=\deriv{H}{z^i}, 
\label{eq:msg1}
\eeqn
where
\beqn
L^\alpha_{i;\alpha}=\deriv{L^\alpha_i}{q^\alpha}
+\Gamma^\alpha_{s\alpha} L^s_i
\equiv\frac{1}{\sqrt{g}}\derv{q^\alpha}\left(\sqrt{g} L^\alpha_i\right), 
\label{eq:msg2}
\eeqn
is the covariant derivative of the contravariant vector field $L^\alpha_i$ and 
$g=\det(g_{\alpha\beta})$ is the determinant of the metric tensor 
where $ds^2=g_{\alpha\beta} dq^\alpha dq^\beta$ is the metric. Here 
$q^\alpha$ are generalized coordinates. The holonomic base vectors 
${\bf e}_\alpha=\partial {\bf x}/\partial q^\alpha$ and 
$g_{\alpha\beta}={\bf e}_\alpha{\bf\cdot}{\bf e}_\beta$. 
In Section 4.1 
we implicitly used Cartesian space time coordinates 
(i.e. $x^\alpha=(t,x,y,z)$). The differential equation system (\ref{eq:msg1}) 
also holds for generalized coordinates (e.g. for spherical polar or 
cylindrical space coordinates). 

The system of 
equations (\ref{eq:msg1})-(\ref{eq:msg2}) can be written in the form:
\beqn
\left(\deriv{L^\alpha_j}{z^i}
-\deriv{L^\alpha_i}{z^j}\right)\deriv{z^j}{q^\alpha}
-\frac{1}{\sqrt{g}}\derv{q^\alpha}
\left(\sqrt{g} L^\alpha_i\right)
=\deriv{H}{z^i}, \label{eq:msg2a}
\eeqn 
which highlights the fact that the system depends on the 1-forms 
$\omega^\alpha=L^\alpha_j dz^j$,  the Hamiltonian $H$ and the metric 
$g_{\alpha\beta}$. 

\cite{Bridges10} develop the theory of multi-symplectic systems of the 
form (\ref{eq:msg1}) and its relationship to the variational 
bi-complex (e.g. \cite{Anderson89, Anderson92}; 
\cite{Bridges10} and 
references therein). 

Below we provide a discussion of the origin of (\ref{eq:msg1}). 
An alternative approach is to use the variational bi-complex analysis 
of \cite{Bridges10}.

\begin{proposition}
The multi-symplectic partial differential equation system (\ref{eq:msg1}) 
arises as the Euler Lagrange equation of the action:
\beqn
J=\int\int L\sqrt{g}\ d^3q\ dt \label{eq:me1}
\eeqn
where 
\beqn
L=L^\alpha_j z^j_{,\alpha} -H(z), \label{eq:me2}
\eeqn
is the Lagrangian in the original coordinates $(t,x,y,z)$. In particular, 
the stationary point conditions:
\beqn
E_i({\bar L})=\sqrt{g}\left\{ {\sf K}^\alpha_{ij} z^j_{,\alpha}
-L^\alpha_{i;\alpha}-\deriv{H}{z^i}\right\}=0, \label{eq:me3}
\eeqn
where ${\bar L}=L\sqrt{g}$, are equivalent to the 
multi-symplectic system (\ref{eq:msg1}).

\end{proposition}

\begin{proof}
Taking the variation of the functional (\ref{eq:me1}) and integrating by 
parts gives:
\begin{align}
\delta J=&\int\int \left(\delta\left(L^\alpha_j 
\sqrt{g}\right) z^j_{,\alpha} +L^\alpha_j 
\sqrt{g}\delta z^j_{,\alpha}-\deriv{H}{z^s}\delta z^s\sqrt{g}\right)
\ d^3q\ dt\nonumber\\
=&\int\int\delta z^j
\left(-\frac{1}{\sqrt{g}}D_\alpha\left(L^\alpha_j\sqrt{g}\right)
-\deriv{H}{z^j} +\deriv{L^\alpha_k}{z^j} z^k_{,\alpha}\right)\sqrt{g}\ d^3q\ dt\nonumber\\
&+\int\ d^3q\int\ dt D_\alpha\left(L^\alpha_j\sqrt{g} \delta z^j\right), 
\label{eq:me4}
\end{align}
where $D_\alpha$ denotes the total partial derivative with respect 
to $q^\alpha$. Noting that
\beqn
D_\alpha\left(L^\alpha_j\sqrt{g}\right)
=\derv{q^\alpha}\left(L^\alpha_j\sqrt{g}\right)
+\derv{z^s}\left(L^\alpha_j\sqrt{g}\right) z^s_{,\alpha}, \label{eq:me5}
\eeqn
in (\ref{eq:me4} we obtain:
\begin{align}
\delta J=&\int\int \delta z^j\left\{{\sf K}^\alpha_{js}z^s_{,\alpha} 
-L^\alpha_{j;\alpha}-\deriv{H}{z^s}\right\}\ \sqrt{g}\ d^3q\ dt\nonumber\\
&+\int\int\ D_\alpha\left(L^\alpha_j\sqrt{g} \delta z^j\right)\ d^3q\ dt. \label{eq:me6}
\end{align}
Dropping the surface term in (\ref{eq:me6}) 
and evaluating $\delta J/\delta z^i$ gives the 
Euler-Lagrange equation (\ref{eq:me3}). This completes the proof.
\end{proof}

\begin{proposition}
Let one-form $\omega^\alpha=L^\alpha_j dz^j$ be a contravariant vector field
with respect to the index $\alpha$, then the multi-symplectic 
Legendre transformation (\ref{eq:m11})-(\ref{eq:m12}) is replaced 
by the more general result:
\beqn
\left(L^\alpha_j(z; x) dz^j\right)_{;\alpha}=d \left[L^\alpha_j
(z; x)z^j_{,\alpha}-H(z)\right]=dL, \label{eq:msg8}
\eeqn
where 
\beqn
L=L^\alpha_j
(z; x)z^j_{,\alpha}-H(z), \label{eq:msg9}
\eeqn
is the Lagrangian density and $H(z)$ is the multi-symplectic Hamiltonian 
for the system (\ref{eq:msg1}), which satisfies the structural conservation law:
\beqn
\kappa^\alpha_{;\alpha}=0\quad \hbox{where}\quad 
\kappa^\alpha=d\omega^\alpha=d\left(L^\alpha_j dz^j\right). \label{eq:msg10}
\eeqn
\end{proposition}

\begin{proof}
\ Computing the covariant derivative of $\omega^\alpha$ we obtain:
\beqn
\omega^\alpha_{;\alpha}=\left(L^\alpha_j dz^j\right)_{;\alpha} 
=L^\alpha_{j;\alpha} dz^j +
\left\{ \deriv{L^\alpha_j}{z^k} z^k_{,\alpha} dz^j 
+L^\alpha_j D_\alpha \left(dz^j\right)\right\}, \label{eq:msg11}
\eeqn
where the first term takes into account changes due to changes 
in $x^\alpha$ and the second term in curly brackets takes into account 
the changes in $z^j$ keeping $x^\alpha$ fixed.  From (\ref{eq:msg11}) 
we obtain:
\begin{align}
\omega^\alpha_{;\alpha}=&L^\alpha_{j;\alpha} dz^j
+\deriv{L^\alpha_j}{z^k} z^k_{,\alpha} dz^j
+L^\alpha_j d\left(z^j_{,\alpha}\right)\nonumber\\
=&\left(L^\alpha_{j;\alpha}+\deriv{L^\alpha_j}{z^k} z^k_{,\alpha}\right) dz^j
+\biggl\{d\left(L^\alpha_j z^j_{,\alpha}\right)
-\deriv{L^\alpha_j}{z^k} dz^k z^j_{,\alpha}\biggr\}\nonumber\\
=&\left(\deriv{L^\alpha_j}{z^k} 
-\deriv{L^\alpha_k}{z^j}\right) z^k_{,\alpha} dz^j
+L^\alpha_{j;\alpha} dz^j
+d\left(L^\alpha_j z^j_{,\alpha}\right)\nonumber\\
=&\left({\sf K}^\alpha_{kj} z^k_{,\alpha} +L^\alpha_{j;\alpha}\right) dz^j
+d(L+H). \label{eq:msg12}
\end{align}
From (\ref{eq:msg12}) we obtain:
\beqn
\omega^\alpha_{;\alpha}=-\left({\sf K}^\alpha_{jk} z^k_{,\alpha}-L^\alpha_{j;\alpha}
-\deriv{H}{z^j}\right) dz^j+dL. \label{eq:msg13}
\eeqn
Using (\ref{eq:msg1}), (\ref{eq:msg13}) reduces to:
\beqn
\omega^\alpha_{;\alpha}=\left(L^\alpha_j dz^j\right)_{;\alpha}=dL. 
\label{eq:msg14}
\eeqn
Taking the exterior derivative of (\ref{eq:msg14}) gives:
\beqn
d\omega^\alpha_{;\alpha}=D_\alpha\left(d\omega^\alpha\right)
=D_\alpha\kappa^\alpha=ddL=0. \label{eq:msg15}
\eeqn
Thus, the differential equation system (\ref{eq:msg1}) is multi-symplectic, 
meaning $D_\alpha\kappa^\alpha=0$, where $\kappa^\alpha=d\omega^\alpha$ 
and $\omega^\alpha=L^\alpha_j dz^j$. This completes the proof.
\end{proof}

The form of Noether's first theorem using the generalized coordinates 
$q^\alpha=(t,q^1,q^2,q^3)$ is given below.

\begin{proposition}
If the action:
\beqn
J=\int\int L\ d^3x\ dt=\int\int\ L\sqrt{g}\ d^3q\ dt\equiv \int\int {\bar L}\ d^3q\ dt, \label{eq:msg16}
\eeqn
is invariant under the infinitesimal Lie transformation:
\beqn
z^{'s}=z^s+\epsilon V^s,\quad q^{'\alpha}=q^\alpha+\epsilon V^{q^\alpha}, 
\quad (0\leq\alpha\leq 3,\quad 1\leq s\leq N), \label{eq:msg17}
\eeqn
and under the divergence transformation:
\beqn
{\bar L}'={\bar L}+\epsilon \nabla{\bf\cdot}\bar{\boldsymbol{\Lambda}}, 
\quad \bar{\boldsymbol{\Lambda}}=\sqrt{g}\boldsymbol{\Lambda}, 
\quad {\bar L}=\sqrt{g} L, \label{eq:msg18}
\eeqn
then the Euler Lagrange equations (\ref{eq:me3}) admit the conservation 
law:
\beqn
\frac{1}{\sqrt{g}}\derv{q^\alpha}\left\{\sqrt{g} 
\left(V^{q^\alpha} L+\hat{V}^{z^j} L^\alpha_j+\Lambda^\alpha
\right)\right\}=0.
\eeqn
\end{proposition}
\begin{proof}
\ The proof follows the usual derivation of Noether's theorem (e.g. 
\cite{Bluman89} 
; \cite{Webb05b}). Note that $q^\alpha=(t,q^1,q^2,q^3)$. 
\end{proof}

The completes our discussion of the multi-symplectic system (\ref{eq:msg1}). 

\section{Multi-Symplectic MHD}

In this section we develop multi-symplectic approaches to the MHD 
equations. In Section 5.1, we develop  a multi-symplectic  form of 
the MHD equations using the magnetic induction ${\bf B}$ as the 
basic variable describing the magnetic field and by using 
Clebsch variables. This leads formally to 15 variable set of equations
in the state vector $z=({\bf u}^T,\rho,S,\mu,{\bf B}^T,\boldsymbol{\Gamma}^T, 
\lambda,\beta,\phi)^T$ in which $\mu$, $\boldsymbol{\Gamma}$, $\lambda$, 
$\beta$, $\phi$ are Clebsch variables and $\rho$, ${\bf u}$, ${\bf B}$ 
and $S$ are physical variables. In Section 5.2 we develop a similar 
multi-symplectic formulation of the MHD equations, except 
that the magnetic vector potential ${\bf A}$, in which 
$\alpha={\bf A}{\bf\cdot}d{\bf x}$ is Lie dragged with the background 
flow is used to describe the magnetic field and ${\bf B}=\nabla\times{\bf A}$ 
is the magnetic field induction. In this latter formulation 
 ${\bf Z}=({\bf u}^T,\rho, S, \mu$, 
${\bf A}^T$, $\boldsymbol{\gamma}^T, \lambda,\beta,\phi)^T$ is the state 
vector of the system and $\mu$,$\boldsymbol{\gamma}$, $\lambda$, $\beta$, 
and $\phi$ are Clebsch potentials. 
We show that there is a direct map 
between the state vector ${\bf z}$ of Section 5.1 and the state vector 
${\bf Z}$ of Section 5.2. 

\subsection*{5.1\ Advected Magnetic Flux Case} 
 
In the Clebsch variables approach, the fluid velocity is given 
 by the expression:
\begin{equation}
\rho{\bf u}=\rho\nabla\phi-\beta\nabla S-\lambda\nabla\mu
-(\nabla\times\boldsymbol{\Gamma})\times {\bf B}
-\boldsymbol{\Gamma} (\nabla{\bf\cdot}{\bf B}),  \label{eq:m21}
\end{equation}
In the standard Clebsch variable formulation (Section 3.1),
 in which $t$ is the evolution 
variable, the canonical coordinates are the physical variables 
$(\rho,S,\mu,{\bf B^T})$ and the Lagrange multipliers 
$(\phi,\beta,\lambda,\boldsymbol{\Gamma}^T)$  are the corresponding canonical momenta 
(the role of the canonical momenta and coordinates 
can be interchanged, simply by changing the sign of the Hamiltonian).
 In the multi-symplectic formulation both space 
and time can be thought of as evolution variables. 

In the multi-symplectic approach used in the present analysis, the Clebsch 
variable expansion for the fluid velocity ${\bf u}$ in (\ref{eq:m21}) 
is re-written in the form:
\begin{equation}
\beta\nabla S+\lambda\nabla\mu 
+\boldsymbol{\Gamma}(\nabla{\bf\cdot}{\bf B})
+{\bf B}{\bf\cdot}\nabla\boldsymbol{\Gamma}
-{\bf B}{\bf\cdot}(\nabla\boldsymbol{\Gamma})^T
-\rho\nabla\phi=-\rho {\bf u}\equiv -\frac{\delta\ell}{\delta{\bf u}}, 
\label{eq:m22}
\end{equation}
where
\begin{equation}
\ell=\int_V\left(\frac{1}{2}\rho |{\bf u}|^2-\varepsilon(\rho,S)
-\frac{B^2}{2\mu_0}\right)\ d^3x, \label{eq:m23}
\end{equation}
is the MHD Lagrangian without constraints. 

\begin{proposition}\label{prop5.1}
The evolution equations (\ref{eq:Clebsch5})-(\ref{eq:Clebsch9}) and the Clebsch 
variable equation (\ref{eq:m22}) 
for $-\delta\ell/\delta {\bf u}$ can be written in the 
multi-symplectic form:
\begin{equation}
\left({\sf K}^0\derv{t}+{\sf K}^1 \derv{x}+{\sf K}^2 \derv{y}
+{\sf K}^3 \derv{z}\right){\bf z}=\frac{\delta {\cal H}}{\delta {\bf z}}, 
\label{eq:m24}
\end{equation}
where ${\sf A}$ is a $15\times15$ matrix differential operator. 
In (\ref{eq:m24})
\begin{equation}
{\bf z}=\left({\bf u}^T,\rho,S,\mu,{\bf B}^T, \boldsymbol{\Gamma}^T, 
\lambda,\beta,\phi\right)^T, \label{eq:m25}
\end{equation}
is a 15-dimensional state vector for the system and the 
${\sf K}^\alpha$ ($\alpha=0,1,2,3$) are 
skew-symmetric $15\times 15$ matrices, and 
\begin{equation}
{\cal H}=-\ell\equiv -\int_V\left(\frac{1}{2}\rho |{\bf u}|^2
-\varepsilon(\rho,S)
-\frac{B^2}{2\mu_0}\right)\ d^3x=\int_V H({\bf z}) d^3x \label{eq:m26}
\end{equation}
is the multi-symplectic Hamiltonian functional for the system. 
The functional or variational derivative 
$\delta {\cal H}/\delta {z^s}=\partial H/\partial z^s$ in the present case.  
The skew-symmetric matrices ${\sf K}^{\alpha}$ satisfy 
equations of the form:
\begin{equation}
{\sf K}^\alpha_{ij} dz^i\wedge dz^j=d\omega^\alpha\quad\hbox{where}\quad 
\omega^\alpha=L^\alpha_j dz^j, \label{eq:m27}
\end{equation} 
are symplectic one-forms. For the MHD system, the one-forms $\omega^\alpha$ 
are given by (up the exterior derivative of a scalar function):
\begin{align}
\omega^0=&\phi d\rho+\beta dS+\lambda d\mu
+ \boldsymbol{\Gamma}{\bf\cdot} d{\bf B}, \label{eq:m28}\\
\omega^i=&\left[{\bf u}\left(\beta dS+\lambda d\mu+\phi d\rho\right)+\rho\phi d{\bf u}
+(\boldsymbol{\Gamma}{\bf\cdot}{\bf B})\ d{\bf u}
-{\bf B} (\boldsymbol{\Gamma}{\bf\cdot}d{\bf u})
+{\bf u} (\boldsymbol{\Gamma}{\bf\cdot}d{\bf B}) \right]^i, \label{eq:m29}\\
\equiv&\left[{\bf u}\left(\beta dS+\lambda d\mu-\rho d\phi\right)+d(\rho\phi {\bf u})
+(\boldsymbol{\Gamma}{\bf\cdot}{\bf u})d{\bf B}-(\boldsymbol{\Gamma}\times d{\bf E})\right]^i
\label{eq:m29a}
\end{align}
where $1\leq i\leq 3$ and 
\begin{equation}
{\bf E}=-{\bf u}\times{\bf B}, \label{eq:m30}
\end{equation}
is the electric field in ideal MHD.  The adjoint ${\sf A}^\dagger$ of the matrix differential operator ${\sf A}$ satisfies the equation:
\begin{equation}
\boldsymbol{\psi}^T{\bf\cdot}{\sf A}{\bf z}
=\derv{x^\alpha}\left(\boldsymbol{\psi}^T{\bf\cdot}
{\sf K}^\alpha {\bf z}\right)+{\bf z}^T{\bf\cdot} {\sf A}^\dagger 
\boldsymbol{\psi}, \label{eq:m30a}
\end{equation}
where
\begin{equation}
{\sf A}^\dagger \boldsymbol{\psi}=\derv{x^\alpha}\left({\sf K}^\alpha \boldsymbol{\psi}\right). \label{eq:m30b}
\end{equation}
Note that $\langle \boldsymbol{\psi},{\sf A}{\bf z}\rangle=\langle {\bf z},{\sf A}^\dagger \boldsymbol{\psi}\rangle$, where $\langle,\rangle$ is the 
usual inner product.  

\end{proposition}

\begin{proof}
To derive (\ref{eq:m24})-(\ref{eq:m29}) first note that the Clebsch variable 
equation (\ref{eq:m22}) for $\rho {\bf u}$ and the evolution equations 
(\ref{eq:Clebsch5})-(\ref{eq:Clebsch9}) can be written in the form:
\begin{align}
&\beta\nabla S+\lambda\nabla\mu+\boldsymbol{\Gamma}(\nabla{\bf\cdot}{\bf B})+{\bf B}{\bf\cdot}\nabla\boldsymbol{\Gamma}-{\bf B}{\bf\cdot}
(\nabla\boldsymbol{\Gamma})^T-\rho\nabla\phi=H_{\bf u}, \nonumber\\
&-D_t\phi=H_\rho,\quad -\beta\nabla{\bf\cdot}{\bf u}-D_t\beta=H_S, 
\quad -\lambda\nabla{\bf\cdot}{\bf u}-D_t\lambda=H_\mu, \nonumber\\
&-\boldsymbol{\Gamma}{\bf\cdot}(\nabla{\bf u})^T
-D_t\boldsymbol{\Gamma}=H_{\bf B}, \quad 
{\bf B}(\nabla{\bf\cdot}{\bf u})
-{\bf B}{\bf\cdot}\nabla {\bf u}
+D_t {\bf B}=H_{\boldsymbol{\Gamma}}, \nonumber\\
&D_t\mu=H_\lambda,\quad D_t S=H_\beta,\quad 
\rho\nabla{\bf\cdot}{\bf u}+D_t\rho=H_\phi, \label{eq:m31}
\end{align}
where $D_t=\partial_t+{\bf u}{\bf\cdot}\nabla$ is the Lagrangian time 
derivative and the multi-symplectic Hamiltonian is given by (\ref{eq:m26}). 
In (\ref{eq:m31}) we use the notation $H_\psi\equiv \partial H/\partial\psi$. 

To obtain the matrices ${\sf K}^\alpha$ in (\ref{eq:m24}) write (\ref{eq:m31}) 
in the matrix form:
\begin{equation}
{\sf A} {\bf z}=H_{\bf z}\quad\hbox{where}
\quad {\sf A}={\sf K}^\alpha\derv{x^\alpha},  \label{eq:m32}
\end{equation}
and $(x^0,x^1,x^2,x^3)\equiv (t,x,y,z)$. Note that the equations  involving 
$H_{\bf u}$, $H_{\bf B}$ and $H_{\boldsymbol{\Gamma}}$ each consist of three 
equations, but the other equations involving $H_\rho$, $H_S$, $H_\mu$, 
$H_\lambda$, $H_\beta$ and $H_\phi$ are single equations. The matrix 
differential operator ${\sf A}$
in (\ref{eq:m32}) has the form:
\begin{equation}
{\sf A}=\left(\begin{array}{ccccccccc}
{\sf O}_{3\times 3}&0&\beta\nabla&\lambda\nabla&
\boldsymbol{\Gamma}\nabla{\bf\cdot}&{\bf V}_{\bf B}
&0&0&-\rho\nabla\\
{\sf O}_{1\times 3}&0&0&0&{\sf O}_{1\times 3}&{\sf O}_{1\times 3}&0&0&-D_t\\
-\beta\nabla{\bf\cdot}&0&0&0&{\sf O}_{1\times 3}&{\sf O}_{1\times 3}&0&-D_t&0\\
-\lambda\nabla{\bf\cdot}&0&0&0&{\sf O}_{1\times 3}&{\sf O}_{1\times 3}&
-D_t&0&0\\
-\boldsymbol{\Gamma}{\bf\cdot}(\nabla\circ)^T&0&0&0&{\sf O}_{3\times 3}&
-{\sf I}_{3\times 3}D_t&0&0&0\\
-V_{\bf B}^\dagger&0&0&0&{\sf I}_{3\times 3}D_t
&{\sf O}_{3\times 3}&0&0&0\\
{\sf O}_{1\times 3}&0&0&D_t&{\sf O}_{1\times 3}&{\sf O}_{1\times 3}&0&0&0\\
{\sf O}_{1\times 3}&0&D_t&0&{\sf O}_{1\times 3}&{\sf O}_{1\times 3}&0&0&0\\
\rho\nabla{\bf\cdot}&D_t&0&0&{\sf O}_{1\times 3}&{\sf O}_{1\times 3}&0&0&0
\end{array}
\right). \label{eq:m33}
\end{equation}
where
\begin{equation}
V_{\bf B}={\bf B}{\bf\cdot}\nabla\circ-{\bf B}{\bf\cdot}(\nabla\circ)^T, 
\quad V_{\bf B}^\dagger={\bf B}{\bf\cdot}(\nabla\circ)^T
-{\bf B}\nabla{\bf\cdot}\circ. \label{eq:m34}
\end{equation}
In (\ref{eq:m33}) ${\sf O}_{3\times 3}$ is the zero $3\times 3$ matrix, 
${\sf I}_{3\times 3}$ is the unit $3\times 3$ unit matrix and 
${\sf O}_{1\times 3}$ is the $1\times 3$ zero matrix. The operator:
\begin{equation}
D_t=\derv{t}+{\bf u}{\bf\cdot}\nabla, \label{eq:m35}
\end{equation}
is the Lagrangian time derivative following the flow. Note that 
$V_{\bf B}^\dagger$ is the adjoint of the operator $V_{\bf B}$ 
with respect to the usual inner product $\langle f,g\rangle =\int fg d^3x$
for real functions (further discussion of $V_{\bf B}^\dagger$ is given 
in Appendix B).  

Using (\ref{eq:m33}) the skew symmetric matrices ${\sf K}^{\alpha}_{ij}$ 
have the form:
\begin{equation}
{\sf K}^{\alpha}_{ij}={\sf k}^{\alpha}_{[i,j]}={\sf k}^{\alpha}_{ij}
-{\sf k}^{\alpha}_{ji}. 
\label{eq:m36}
\end{equation}
In particular:
\begin{equation}
{{\sf k}}^{0}_{ij}={\delta}^i_{15} {\delta}^j_{4}
+{\delta}^i_{14} {\delta}^j_{5}
+{\delta}^i_{13} {\delta}^j_{6}+{\delta}^i_{10} {\delta}^j_{7}
+{\delta}^i_{11} 
{\delta}^j_{8}+{\delta}^i_{12} {\delta}^j_{9}. \label{eq:m37}
\end{equation}
Similarly:
\begin{align}
{\sf k}^1_{ij}=& \Gamma^x\delta^i_1\delta^j_7+\Gamma^y\delta^i_2\delta^j_7
+\Gamma^z\delta^i_3\delta^j_7\nonumber\\
& +B^x\left(\delta^i_2\delta^j_{11}+\delta^i_3\delta^j_{12}\right)
+B^y\delta^i_{11}\delta^j_1+B^z\delta^i_{12}\delta^j_{1}
+u^x\left(\delta^i_{10}\delta^j_7+\delta^i_{11}\delta^j_8
+\delta^i_{12}\delta^j_9\right)\nonumber\\
&+\left\{ u^x\left(\delta^i_{14}\delta^j_5+\delta^i_{13}\delta^j_6+
\delta^i_{15}\delta^j_4\right)
+\beta\delta^i_1\delta^j_5+\lambda\delta^i_1\delta^j_6
-\rho\delta^i_1\delta^j_{15}\right\}, \label{eq:m38}\\
{\sf k}^2_{ij}=&\Gamma^x \delta^i_1\delta^j_8+\Gamma^y\delta^i_2\delta^j_8
+\Gamma^z\delta^i_3\delta^j_8\nonumber\\
&+B^x\delta^i_{10}\delta^j_2
+B^y\left(\delta^i_1\delta^j_{10}+\delta^i_3 \delta^j_{12}\right)
+B^z\delta^i_{12}\delta^j_2
+u^y\left(\delta^i_{10}\delta^j_7+\delta^i_{11}\delta^j_8
+\delta^i_{12}\delta^j_9\right)\nonumber\\
&+\left\{u^y\left(\delta^i_{13}\delta^j_6+\delta^i_{14}\delta^j_5+\delta^i_{15}\delta^j_4\right) 
+\beta\delta^i_2\delta^j_5+\lambda\delta^i_2\delta^j_6
-\rho\delta^i_2\delta^j_{15}\right\}, \label{eq:m39}\\
k^3_{ij}=&\Gamma^x \delta^i_1\delta^j_9+\Gamma^y\delta^i_2\delta^j_9
+\Gamma^z\delta^i_3\delta^j_9\nonumber\\
&+B^x\delta^i_{10}\delta^j_3+B^y\delta^i_{11}\delta^j_3
+B^z\left(\delta^i_1\delta^j_{10}+\delta^i_2\delta^j_{11}\right)
+u^z\left(\delta^i_{10}\delta^j_7+\delta^i_{11}\delta^j_8
+\delta^i_{12}\delta^j_9\right)\nonumber\\
&+\left\{u^z\left(\delta^i_{13}\delta^j_6+\delta^i_{14}\delta^j_5
+\delta^i_{15}\delta^j_4\right) 
+\beta\delta^i_3\delta^j_5+\lambda\delta^i_3\delta^j_6
-\rho\delta^i_3\delta^j_{15}\right\}. \label{eq:m40}
\end{align}

The one-form solutions for $\omega^\alpha=L_j^\alpha dz^j$ 
in (\ref{eq:m28})-(\ref{eq:m29}) are related to the ${\sf K}^\alpha_{jk}$ by 
(\ref{eq:m10}), i.e. 
\begin{equation}
{\sf K}_{jk}^\alpha=\deriv{L^\alpha_k}{z^j}-\deriv{L_j^\alpha}{z^k}. 
\label{eq:m41}
\end{equation}
Note that the solution of (\ref{eq:m41}) for the $L_j^\alpha$ 
are not unique because 
$\omega^\alpha=L_j^\alpha dz^j+d\Phi({\bf z})^\alpha$ will also give the same 
${\sf K}^\alpha_{jk}$. 

As an example we find $\omega^0=L^0_j dz^j$ is given by:
\begin{align}
\omega^0=&\left(z^{15}dz^4+z^{14}dz^5+z^{13} dz^6\right) 
+\left\{z^{10}dz^7+z^{11}dz^8+z^{12}dz^9\right\}\nonumber\\
\equiv& \phi d\rho+\beta dS+\lambda d\mu
+\boldsymbol{\Gamma}{\bf\cdot}d{\bf B}.
\label{eq:m42} 
\end{align}
Similarly, we obtain (\ref{eq:m29}) for $\omega^i$. 

\end{proof}
\subsubsection*{5.1.1\ Exterior differential forms approach}
Although the above derivation of the MHD multi-symplectic structure is straightforward, there 
is some ambiguity in the one forms $\omega^0$ and $\omega^i$ in (\ref{eq:m28})-(\ref{eq:m29}) 
since one can always add a perfect differential to these forms. A more elegant way to derive the 
above results of Proposition \ref{prop5.1}  is to use differential forms to deduce 
the skew symmetric matrices ${\sf K}^\alpha$ and the one forms $\omega^\alpha$ ($\alpha=0,1,2,3$)
describing the system. This approach is described below. 
From (\ref{eq:Clebsch2}) the MHD Lagrangian may be written in the form:
\beqn
L=\frac{1}{2}\rho u^2-\varepsilon(\rho,S)-\frac{B^2}{2\mu} 
+L^{\alpha}_{z^s} \deriv{z^s}{x^\alpha}, 
\label{eq:m42a}
\eeqn
where
\begin{align}
L^{\alpha}_{z^s} \deriv{z^s}{x^\alpha}=&\phi\left(\deriv{\rho}{t}+\nabla{\bf\cdot}(\rho {\bf u})\right)
+\beta\left(\deriv{S}{t}+{\bf u}{\bf\cdot}\nabla S\right)
+\lambda\left(\deriv{\mu}{t}+{\bf u\cdot}\nabla\mu\right) \nonumber\\
&+\boldsymbol{\Gamma}{\bf\cdot}\left(\deriv{\bf B}{t}-\nabla\times({\bf u}\times{\bf B})
+{\bf u}(\nabla{\bf\cdot B})\right). \label{eq:m42b}
\end{align}
In particular:
\begin{equation}
L^0_{z^s} \deriv{z^s}{x^0}=\phi\rho_t+\beta S_t+\lambda\mu_t+\boldsymbol{\Gamma}{\bf\cdot}{\bf B}_t, 
\label{eq:m42c}
\end{equation}
and hence 
\begin{align}
\omega^0=&L^0_\rho d\rho+L_S^0 dS+L_\mu^0 d\mu+L^0_{\bf B}{\bf\cdot} d{\bf B}, \nonumber\\
\equiv&\phi d\rho +\beta dS+\lambda d\mu+\boldsymbol{\Gamma}{\bf\cdot}d{\bf B}, \label{eq:m42d}
\end{align}
(note (\ref{eq:m42d}) define the non-zero $L^0_{z^s}$). The result (\ref{eq:m42d}) for $\omega^0$ 
is the same as (\ref{eq:m28}). 
Taking the exterior derivative of (\ref{eq:m42d}) gives:
\beqn
d\omega^0=d\phi\wedge d\rho+d\beta\wedge dS+d\lambda\wedge d\mu+d\Gamma_s\wedge dB^s
\equiv \frac{1}{2}K^0_{z^s,z^p}dz^s\wedge dz^p. \label{eq:m42e}
\eeqn
Hence
\beqn
{\sf K}^0_{\phi,\rho}={\sf K}^0_{\beta,S}={\sf K}^0_{\lambda,\mu}={\sf K}^0_{\Gamma_s,B^s}=1
\quad (s=1,2,3). \label{eq:m42f}
\eeqn
Thus we obtain the skew symmetric matrix ${\sf K}^0_{ij}$ given by (\ref{eq:m36}) and (\ref{eq:m37}). 

A similar calculation gives:
\begin{align}
L^k_{z^s} \deriv{z^s}{x^k}=&\phi\left(\rho\nabla_k u^k+u^k\nabla_k\rho\right) 
+\beta \left(u^k\nabla_k S\right)+\lambda u^k \nabla_k\mu\nonumber\\
&+\Gamma_s\left(u^k\nabla_k B^s+B^s\nabla_k u^k-B^k\nabla_k u^s\right), \label{eq:m42g}
\end{align}
from which we read off:
\begin{align}
L^k_\rho=&\phi u^k, \quad L^k_{u^i}=(\rho\phi+\boldsymbol{\Gamma}{\bf\cdot}{\bf B}) \delta^k_i
-\Gamma_i B^k, 
\nonumber\\
L^k_S=&\beta u^k,\quad L^k_\mu=\lambda u^k, \quad L^k_{B^i}=\Gamma_i u^k, \label{eq:m42h}
\end{align}
Using (\ref{eq:m42h}) we obtain:
\begin{equation}
\omega^k=L^k_{z^s} dz^s=\left\{{\bf u}[\phi\ d\rho+\beta dS+\lambda d\mu]  +\rho \phi d{\bf u}
+(\boldsymbol{\Gamma}{\bf\cdot}{\bf B})d{\bf u}-{\bf B}(\boldsymbol{\Gamma}{\bf\cdot}d{\bf u})
+{\bf u}(\boldsymbol{\Gamma}{\bf\cdot}d{\bf B})\right\}^k, \label{eq:m42i}
\end{equation}
which is the result (\ref{eq:m29}) for $\omega^k$. 
Taking the exterior derivative of (\ref{eq:m42i}) gives:
\begin{align}
d\omega^k=&du^k\wedge\left(\beta dS+\lambda d\mu-\rho d\phi-B^s d\Gamma_s\right)\nonumber\\
&+u^k \left(d\phi\wedge d\rho+d\beta\wedge dS
+d\lambda\wedge d\mu+d\Gamma_s\wedge dB^s\right)\nonumber\\
&-\Gamma_s dB^k\wedge du^s- B^k d\Gamma_s\wedge du^s. \label{eq:m42j}
\end{align}
From (\ref{eq:m42j}) we obtain:
\begin{align}
&{\sf K}^k_{u^k,S}=\beta,\quad {\sf K}^k_{u^k,\mu}=\lambda,
\quad {\sf K}^k_{u^k,\phi}=-\rho,\nonumber\\
&{\sf K}^k_{\Gamma_s,B^s}= 
{\sf K}^k_{\phi,\rho}={\sf K}^k_{\beta,S}={\sf K}^k_{\lambda,\mu}=u^k, \quad
{\sf K}^k_{u^s,B^k}=\Gamma_s, \nonumber\\
&{\sf K}^k_{u^k,\Gamma_s}=-B^s,\quad {\sf K}^k_{u^s,\Gamma_s}=B^k\quad (k\neq s). \label{eq:m42k}
\end{align}
By using the state vector ${\bf z}=\left({\bf u}^T,\rho,S,\mu,{\bf B}^T,\boldsymbol{\Gamma}^T,
\lambda,\beta,\phi\right)^T$ and (\ref{eq:m42k}) gives the results (\ref{eq:m38})-(\ref{eq:m40}) 
for the ${\sf K}^k_{ij}$ 
($k=1,2,3$).

\subsubsection*{5.1.2\ Multi-Symplectic Conservation Laws}
\begin{proposition}\label{prop5.2}
Using the results (\ref{eq:m28}) and (\ref{eq:m29a}) for the one-forms $\omega^0$  and 
$\omega^k$ ($k=1,2,3$), the multi-symplectic conservation law (\ref{eq:m18}) for $\beta=0$  
reduces to:
\begin{equation}
\deriv{D}{t}+\nabla{\bf\cdot} {\bf F}=0, \label{eq:m47}
\end{equation}
where
\begin{align}
D=&\left(\frac{1}{2}\rho |{\bf u}|^2
+\varepsilon(\rho,S)+\frac{B^2}{2\mu_0}\right)
-\nabla{\bf\cdot}({\bf E}\times\boldsymbol{\Gamma}+\rho\phi {\bf u}), \nonumber\\
{\bf F}=&{\bf u}\left(\frac{1}{2}\rho |{\bf u}|^2+\varepsilon(\rho,S)+p\right)
+\frac{{\bf E}\times{\bf B}}{\mu_0}
+\derv{t}({\bf E}\times\boldsymbol{\Gamma}+\rho\phi{\bf u}) 
-\nabla\times\left[(\boldsymbol{\Gamma}{\bf\cdot}{\bf u}){\bf E}
\right]. \label{eq:m48}
\end{align}
Because of null divergence terms in (\ref{eq:m48}), the conservation law 
(\ref{eq:m47}) reduces to the MHD energy conservation equation:
\begin{equation}
\derv{t}\left(\frac{1}{2}\rho |{\bf u}|^2+\varepsilon(\rho,S)
+\frac{B^2}{2\mu_0}\right) 
+\nabla{\bf\cdot} 
\left({\bf u}\left(\frac{1}{2}\rho |{\bf u}|^2
+\varepsilon(\rho,S)+p\right)
+\frac{{\bf E}\times{\bf B}}{\mu_0}\right)=0. \label{eq:m49}
\end{equation}

Similarly, the multi-symplectic conservation law (\ref{eq:m18}) for $\beta=k$ 
gives a conservation law of the form (\ref{eq:m47}) but with 
\begin{align}
D\equiv D^k=&-\rho u^k+\nabla_k(\rho\phi+\boldsymbol{\Gamma}{\bf\cdot}{\bf B})
-\nabla{\bf\cdot}(\Gamma^k {\bf B}), \nonumber\\
F^i\equiv F^{ik}=&-\left\{\rho u^iu^k
+\left(p+\frac{B^2}{2\mu_0}\right)\delta^{ik}
-\frac{B^iB^k}{\mu_0}\right\}\nonumber\\
&+\left[-\derv{t}(\rho\phi+\boldsymbol{\Gamma}{\bf\cdot}{\bf B}) \delta^{ik} 
+\derv{t}\left(\Gamma^k B^i\right)\right]\nonumber\\
&+\nabla\times\left(\Gamma^k{\bf E}\right)^i
+\nabla_k[\boldsymbol{\Gamma}{\bf\cdot}{\bf B} u^i] 
-\nabla{\bf\cdot}\left[\boldsymbol{\Gamma}{\bf\cdot}{\bf B}{\bf u}\right]
\delta^{ik}. \label{eq:m50}
\end{align}
The conservation law (\ref{eq:m47}) reduces to:
\begin{equation}
-\left\{\derv{t}(\rho {\bf u})
+\nabla{\bf\cdot}\left[\rho {\bf u}\otimes {\bf u}
+\left(p+\frac{B^2}{2\mu_0}\right)
\sf{I}-\frac{{\bf B}\otimes{\bf B}}{\mu_0}\right]\right\}^k=0, \label{eq:m51}
\end{equation}
i.e., the conservation law reduces to the MHD momentum conservation equation 
in the $x^k$- direction. 
\end{proposition}

\begin{proof}
The multi-symplectic Hamiltonian density $H$, and the 1-forms 
$\omega^\alpha=\Gamma^\alpha_j dz^j$ from (\ref{eq:m28}) 
and (\ref{eq:m29}) give:
\begin{align}
L^0_j z^j_{,\alpha}=&\phi\rho_{,\alpha}+\beta S_{,\alpha}
+\lambda\mu_{,\alpha}+\Gamma_s B^s_{,\alpha}, \nonumber\\
L^i_j z^j_{,\alpha}=&u^i\left[\beta S_{,\alpha}+\lambda\mu_{,\alpha}
-\rho\phi_{,\alpha}\right]
+\boldsymbol{\Gamma}{\bf\cdot}{\bf u} B^i_{,\alpha}+(\rho\phi u^i)_{,\alpha}
-\epsilon_{ijk} \Gamma^j E^k_{,\alpha}\quad (\alpha=0,1,2,3).  
\label{eq:m52}
\end{align}
Using (\ref{eq:m52}) we obtain:
\begin{align}
L=&L^\alpha_j z^j_{,\alpha}-H\equiv p-\frac{B^2}{2\mu_0}+\derv{t}(\rho\phi)
+\nabla{\bf\cdot}(\rho\phi {\bf u}), 
\nonumber\\ 
H=&-\left(\frac{1}{2}\rho |{\bf u}|^2-\varepsilon(\rho,S)
-\frac{B^2}{2\mu_0}\right).
\label{eq:m53}
\end{align}
Using the results (\ref{eq:m52})-(\ref{eq:m53}) in the symplectic 
conservation law (\ref{eq:m18}) for $\beta=0$ and $\beta=k$ gives the energy 
and momentum conservation laws  (\ref{eq:m49}) and (\ref{eq:m51}).
\end{proof}

\begin{proposition}\label{prop5.3}
{\bf The symplecticity or structural conservation laws (\ref{eq:m19}) 
for MHD and gas dynamics can be written in the form:
\begin{equation}
D_\alpha\left(F^\alpha_{ab}\right)=0,\quad a<b \label{eq:m54}
\end{equation}
where
\begin{equation}
F^\alpha_{ab}={\bf z}^T_{,a} {\sf K}^\alpha {\bf z}_{,b}.
\label{eq:m55}
\end{equation}
 $F^\alpha_{ab}$ can be calculated by noting that:
\begin{equation}
d\omega^\alpha=F^\alpha_{ab} dx^a\otimes dx^b\quad\hbox{where}\quad a<b.
 \label{eq:m56}
\end{equation}
Using (\ref{eq:m42d}) and (\ref{eq:m42i}) for $\omega^0$ and $\omega^i$ 
($i=1,2,3$) and calculating  $d\omega^0$ and $d\omega^i$ gives the 
formulae:
\begin{align}
F^0_{ab}=&\frac{\partial(\phi,\rho)}{\partial(x^a,x^b)}
+\frac{\partial(\beta,S)}{\partial(x^a,x^b)}
+\frac{\partial(\lambda,\mu)}{\partial(x^a,x^b)}
+\frac{\partial(\Gamma_s,B^s)}{\partial(x^a,x^b)}, \label{eq:m57}\\
F^i_{ab}=&\frac{\partial(\phi u^i,\rho)}{\partial(x^a,x^b)}
+\frac{\partial(\beta u^i,S)}{\partial(x^a,x^b)}
+\frac{\partial(\lambda u^i,\mu)}{\partial(x^a,x^b)}\nonumber\\
&+\frac{\partial(\rho\phi,u^i)}{\partial(x^a,x^b)}
+ \frac{\partial(u^i\Gamma_s,B^s)}{\partial(x^a,x^b)}
-\frac{\partial(B^i\Gamma_s,u^s)}{\partial(x^a,x^b)}. \label{eq:m58}
\end{align}
}
\end{proposition}
\begin{proof}

{\bf We give the derivation of (\ref{eq:m57}).  
 Using (\ref{eq:m42e}) we obtain:
\begin{align}
d\omega^0=&\left(\deriv{\phi}{x^a}\deriv{\rho}{x^b}+\deriv{\beta}{x^a}\deriv{S}{x^b} 
+\deriv{\lambda}{x^a}\deriv{\mu}{x^b}
+\deriv{\Gamma_s}{x^a}\deriv{B^s}{x^b}\right) dx^a\wedge dx^b\nonumber\\
=&\biggl\{ \frac{\partial(\phi,\rho)}{\partial(x^a,x^b)}
+\frac{\partial(\beta,S)}{\partial(x^a,x^b)} 
+\frac{\partial(\lambda,\mu)}{\partial(x^a,x^b)}
+\frac{\partial(\Gamma_s,B^s)}{\partial(x^a,x^b)}\biggr\} 
dx^a\otimes dx^b. \label{eq:m59}
\end{align}
Using (\ref{eq:m59}) we obtain the formula (\ref{eq:m57}) for $F^0_{ab}$.  

A similar analysis gives:
\begin{equation}
d\omega^i=F^i_{ab} dx^a\otimes dx^b, \quad (i=1,2,3), \label{eq:m60}
\end{equation}
where $F^i_{ab}$ is given by (\ref{eq:m58}). This completes the proof.}
\end{proof}

\subsection*{5.2\ Advected ${\bf A}{\bf\cdot}d{\bf x}$ formulation}
In Section 3.2, we discussed the MHD variational equations that result 
when the condition that the one-form $\alpha={\bf A}{\bf\cdot}d{\bf x}$ 
is advected with the background flow is used instead of Faraday's 
equation which is equivalent to the condition that the magnetic flux 
2-form $\beta={\bf B}{\bf\cdot}d {\bf S}$ is Lie dragged with the 
flow, for the case where $\nabla{\bf\cdot}{\bf B}=0$ and 
${\bf B}=\nabla\times{\bf A}$ (e.g., \cite{Tur93}, 
\cite{Webb14a}). Note that 
$\beta=d\alpha={\bf B}{\bf\cdot}d{\bf S}$ in the above case. In addition
\beqn
H=\int_V \alpha\wedge\beta=\int_V({\bf A}{\bf\cdot}d{\bf x})\wedge 
({\bf B}{\bf\cdot}d{\bf S})=\int_V {\bf A}{\bf\cdot}{\bf B}\ d^3x, 
\label{eq:ada1}  
\eeqn
is the Hopf invariant or magnetic helicity (note if there is not a global 
 ${\bf A}$ with ${\bf B}=\nabla\times{\bf A}$, then the magnetic field 
has a non-trivial topology, e.g. \cite{Moffatt69}, \cite{Arnold98}).

\begin{proposition}\label{prop5.4}
Choose the gauge of ${\bf A}$ such that $\alpha={\bf A}{\bf\cdot}d{\bf x}$, 
is advected with the background flow. The Clebsch variational equations 
(\ref{eq:can3})-(\ref{eq:can5}) imply:
\begin{align}
&\beta\nabla S+\lambda\nabla\mu+\boldsymbol{\gamma}{\bf\cdot}(\nabla{\bf A})^T
-\boldsymbol{\gamma}{\bf\cdot}\nabla{\bf A}
-{\bf A}(\nabla{\bf\cdot}\boldsymbol{\gamma}) 
-\rho\nabla\phi={\cal H}_{\bf u}\equiv-\rho{\bf u}, \nonumber\\
&-D_t\phi={\cal H}_\rho\equiv-\left(\frac{1}{2}|{\bf u}|^2-h\right), 
\quad -\beta\nabla{\bf\cdot}{\bf u}-D_t(\beta)={\cal H}_S\equiv \rho T, 
\nonumber\\
&-\lambda\nabla{\bf\cdot}{\bf u}-D_t\lambda={\cal H}_\mu\equiv 0, \quad
-\boldsymbol{\gamma}\nabla{\bf\cdot}{\bf u}
+(\boldsymbol{\gamma}{\bf\cdot}\nabla){\bf u}
-D_t\boldsymbol{\gamma}={\cal H}_{\bf A}\equiv {\bf J}, \nonumber\\
&{\bf A}{\bf\cdot}(\nabla{\bf u})^T+D_t{\bf A}
={\cal H}_{\boldsymbol{\gamma}}\equiv 0,\quad
D_t\mu={\cal H}_\lambda=0, \quad\nonumber\\
&D_tS={\cal H}_\beta\equiv 0, \quad
\rho \nabla{\bf\cdot}{\bf u}+D_t\rho={\cal H}_\phi\equiv 0, 
\label{eq:ada2}
\end{align}
where ${\bf J}=\nabla\times{\bf B}/\mu_0$ is the MHD currrent, 
\beqn
{\cal H}=-\ell=-\int_V\biggl(\frac{1}{2}\rho |{\bf u}|^2
-\varepsilon (\rho,S)
-\frac{|\nabla\times{\bf A}|^2}{2\mu_0}
\biggr)\ d^3x\equiv \int_V H({\bf Z})\ d^3x, \label{eq:ada3}
\eeqn
is the multi-symplectic Hamiltonian and
\beqn
{\bf Z}=\left({\bf u}^T,\rho,S,\mu,{\bf A}^T,\boldsymbol{\gamma}^T,
\lambda,\beta,\phi\right)^T, \label{eq:ada4}
\eeqn
is the 15-dimensional state vector for the system. 
The transformations:
\beqn
\boldsymbol{\Gamma}^*=-{\bf A},\quad {\bf B}^*
=\boldsymbol{\gamma}, \label{eq:ada5}
\eeqn
formally maps the left handsides of the 
advected ${\bf A}$ variational equations (\ref{eq:can3})-(\ref{eq:can5}) 
(i.e. (\ref{eq:ada2})) onto the Clebsch  variational equations (\ref{eq:m31}) 
associated with the advection of the magnetic 
flux ${\bf B}{\bf\cdot}d{\bf S}$, where to obtain (\ref{eq:m31}) 
the replacements 
$\boldsymbol{\Gamma}^*\to \boldsymbol{\Gamma}$ and 
${\bf B}^*\to {\bf B}$ are made.
\end{proposition}

\begin{proof}
To verify that the map (\ref{eq:ada5}) maps the Clebsch 
equations (\ref{eq:ada2}) associated with the Lie dragging 
of ${\bf A}{\bf\cdot}d{\bf x}$ onto the Clebsch equations (\ref{eq:m31}),
it suffices to consider only the equations related to 
the evolution of ${\bf A}$ and $\boldsymbol{\gamma}$  
and of ${\bf B}$ and $\boldsymbol{\Gamma}$. 
The first equation in (\ref{eq:ada2}) under the map (\ref{eq:ada5}) 
becomes:
\beqn
\beta\nabla S+\lambda\nabla\mu
-{\bf B}^*{\bf\cdot}\left(\nabla\boldsymbol{\Gamma}^*\right)^T
+({\bf B}^*{\bf\cdot}\nabla)\boldsymbol{\Gamma}^*
+\boldsymbol{\Gamma}^*\left(\nabla{\bf\cdot}{\bf B}^*\right)
-\rho\nabla\phi={\cal H}_{\bf u}, \label{eq:ada6}
\eeqn
which is the first equation in (\ref{eq:m31}), but with ${\bf B}\to {\bf B}^*$ 
and $\boldsymbol{\Gamma}\to\boldsymbol{\Gamma}^*$. 

Similarly, the $5^{th}$ equation in (\ref{eq:ada2}) becomes:
\beqn
-\left\{{\bf B}^* (\nabla{\bf\cdot}{\bf u})-({\bf B}^*
{\bf\cdot}\nabla) {\bf u}
+D_t{\bf B}^*\right\}=-{\cal H}_{\boldsymbol{\Gamma}^*}, \label{eq:ada7}
\eeqn
which is equivalent to the sixth equation in (\ref{eq:m31}). 

The sixth equation in (\ref{eq:ada2}) becomes:
\beqn
-\boldsymbol{\Gamma}^*{\bf\cdot}(\nabla {\bf u})^T -D_t \boldsymbol{\Gamma}^* 
={\cal H}_{{\bf B}^*}, \label{eq:ada8}
\eeqn
which is the $5^{th}$ equation in (\ref{eq:m31}), but with 
$\boldsymbol{\Gamma}\to \boldsymbol{\Gamma}^*$ and ${\bf B}\to {\bf B}^*$. 
This completes the proof.
\end{proof}

\begin{remark}
The Hamiltonian for the system (\ref{eq:ada2}) is given by (\ref{eq:ada3}). Under the 
transformations (\ref{eq:ada5}) the Hamiltonian functional ${\cal H}$ 
becomes:
\beqn
{\cal H}=-\int_V\left(\frac{1}{2}\rho |{\bf u}|^2-\varepsilon(\rho,S)
-\frac{|\nabla\times\boldsymbol{\Gamma}^*|^2}{2\mu_0}\right)\ d^3x. \label{eq:ada9}
\eeqn
\end{remark}

{\bf 
\begin{proposition}
The one forms describing the multi-symplectic MHD system, using the advected ${\bf A}{\bf\cdot}d{\bf x}$ 
formalism are:
\begin{align}
\omega^0=&\phi d\rho+\beta dS+\lambda d\mu+\gamma_s dA^s, \label{eq:ada10}\\
\omega^k=&u^k(\beta dS+\lambda d\mu+\phi d\rho)+\phi\rho du^k
+\gamma_k A^s du^s+u^k \gamma_s dA^s,\quad k=1,2,3, \label{eq:ada11}
\end{align}
The forms (\ref{eq:ada10}) and (\ref{eq:ada11}) can then be used to determine 
the skew symmetric matrices ${\sf K}^\alpha$, for $\alpha=0,1,2,3$. 
\end{proposition}
}
\begin{proof}
Using the action principle (\ref{eq:can2}) we can write:
\beqn 
L=\frac{1}{2}\rho u^2-\varepsilon(\rho,S)-\frac{|\nabla\times{\bf A}|^2}{2\mu}
+ L^\alpha_{Z^s}\deriv{Z^s}{x^\alpha}, \label{eq:prla1}
\eeqn
where
\begin{align}
L^\alpha_{Z^s} \deriv{Z^s}{x^\alpha}=&\phi\left[\rho_t+\nabla(\rho {\bf u})\right]
+\beta(S_t+{\bf u}{\bf\cdot}\nabla S)+\lambda(\mu_t+{\bf u}{\bf\cdot}\nabla\mu)\nonumber\\
&+\boldsymbol{\gamma}{\bf\cdot}\left[{\bf A}_t-{\bf u}\times 
(\nabla\times{\bf A}) +\nabla({\bf u}{\bf\cdot}{\bf A})\right]. \label{eq:prla2}
\end{align}
Following the approach in (\ref{eq:m42a}) et seq., we identify:
\begin{align}
\omega^0=&L^0_\rho d\rho+L^0_S dS+L^0_\mu d\mu +L^0_{\bf A}{\bf\cdot} d{\bf A}, \nonumber\\
\equiv& \phi d\rho+\beta dS+\lambda d\mu+\gamma_s dA^s, \label{eq:prla3}
\end{align}
and hence:
\beqn
L^0_\rho=\phi,\quad L^0_S=\beta,\quad L^0_\mu=\lambda,\quad L^0_{A^s}=\gamma_s. \label{eq:prla4}
\eeqn
Taking the exterior derivative of (\ref{eq:prla3}) gives:
\beqn
d\omega^0=d\phi\wedge d\rho+d\beta\wedge dS+d\lambda\wedge d\mu+d\gamma_s\wedge dA^s, \label{prla5}
\eeqn
which implies:
\beqn
{\sf K}^0_{\phi,\rho}={\sf K}^0_{\beta,S}={\sf K}^0_{\lambda,\mu}={\sf K}^0_{\gamma_s,A^s}=1. 
\label{eq:prla6}
\eeqn

Similarly using (\ref{eq:prla2}) we obtain:
\beqn
\omega^k=L^k_{Z^s} dZ^s=u^k(\beta dS+\lambda d\mu+\phi d\rho)+\phi\rho du^k
+\gamma_k A^s du^s+u^k \gamma_s dA^s, \label{eq:prla7}
\eeqn
and hence:
\begin{align}
&L^k_S=\beta u^k,\quad L^k_\mu=\lambda u^k,\quad L^k_\rho= \phi u^k,
\quad L^k_{u^k}=\phi\rho+\gamma_k A^k, \nonumber\\
&L^k_{u^s}=\gamma_k A^s,\quad L^k_{A^s}= u^k \gamma_s. \label{eq:prla8}
\end{align}
Taking the exterior derivative of (\ref{eq:prla7}) gives:
\begin{align}
d\omega^k=&du^k[\beta dS+\lambda d\mu+ \phi d\rho]
+u^k[d\beta\wedge dS+d\lambda\wedge d\mu+ d\phi\wedge d\rho]
+(\phi d\rho+\rho d\phi)\wedge du^k \nonumber\\
&+\left(d\gamma_k A^s+\gamma_k dA^s\right)\wedge du^s
+\left(du^k\gamma_s+ u^k d\gamma_s\right)\wedge dA^s, \label{eq:prla9}
\end{align}
which can be used to determine the ${\sf K}^k_{Z^\mu,Z^\nu}$, i.e.
\begin{align}
&{\sf K}^k_{u^k,S}=\beta,\quad {\sf K}^k_{u^k,\mu}=\lambda, \nonumber\\
&{\sf K}^k_{\beta,S}={\sf K}^k_{\lambda,\mu}={\sf K}^k_{\phi,\rho}=u^k, 
\quad {\sf K}^k_{\phi,u^k}=\rho, \nonumber\\
&{\sf K}^k_{A^s,u^s}=\gamma_k, \quad {\sf K}^k_{u^k,A^s}=\gamma_s, 
\quad (k\neq s),\nonumber\\
&{\sf K}^k_{\gamma_k,u^s}=A^s,\quad {\sf K}^k_{\gamma_k,A^s}=u^k. \label{eq:prla10}
\end{align}
\end{proof}

{\bf
\begin{proposition}\label{prop5.5}
The structural conservation laws for the advected ${\bf A}{\bf\cdot}d{\bf x}$ version of 
the MHD equations is given 
by (\ref{eq:m54}) in which:
\begin{align}
F^0_{ab}=&\frac{\partial(\phi,\rho)}{\partial(x^a,x^b)} 
+\frac{\partial(\beta,S)}{\partial(x^a,x^b)}
+\frac{\partial(\lambda,\mu)}{\partial(x^a,x^b)}
+\frac{\partial(\gamma_s,A^s)}{\partial(x^a,x^b)}, \label{eq:prla11}\\
F^i_{ab}=&\frac{\partial(\phi u^i,\rho)}{\partial(x^a,x^b)}
+\frac{\partial(\beta u^i,S)}{\partial(x^a,x^b)}
+\frac{\partial(\lambda u^i,\mu)}{\partial(x^a,x^b)}\nonumber\\
&+\frac{\partial(\rho \phi,u^i)}{\partial(x^a,x^b)}
+\frac{\partial(\gamma_i A^s,u^s)}{\partial(x^a,x^b)}
+\frac{\partial(\gamma_s A^i,A^s)}{\partial(x^a,x^b)}, \label{eq:prla12}
\end{align}
($i=1,2,3$ and $a,b=0,1,2,3$) 
are the conserved densities and fluxes respectively where $a<b$. 
\end{proposition}
\begin{proof}
The proof is the same as in Proposition \ref{prop5.3}, except that $\omega^0$ and $\omega^i$ 
are now given by (\ref{eq:ada10}) and (\ref{eq:ada11}).
\end{proof}
}

\section{Summary and Concluding Remarks}
In this paper we developed multi-symplectic equations for ideal MHD. 
A key ingredient was the use of Clebsch variable variational 
principles in which the constraint equations 
(mass continuity, entropy advection, Lin constraint, Faraday's equation, 
or its analogue for the magnetic vector potential ${\bf A}$), are ensured by 
using Lagrange multipliers. The connection between Clebsch variables 
and the momentum map in ideal fluid systems with constraints 
was used (Section 3.1). The Lin constraint can be viewed as defining a Lagrangian 
variable which is advected with the flow (more generally it is useful 
to include 3 Lagrangian labels, to include different possible initial 
data, that are not included in the usual Clebsch variable description (see
e.g. \cite{Cotter07}). 

{\bf Section 2 introduces the MHD equations. In Section 2 we also discuss 
one dimensional gas dynamics as an example of a multi-symplectic 
system. The example illustrates, that both time and space can be 
thought of as evolution variables. The example uses Clebsch variables
to describe the gas dynamic equations. The multi-symplectic formulation 
involves two skew symmetric matrices associated with the time and space 
evolution. The gas dynamic equations are obtained by finding the stationary point 
conditions for the action, including the mass continuity equation and entropy 
advection equation constraints by using Lagrange multipliers. 
 One forms  (i.e. differential forms) are constructed from the constraint equations 
which lead to the multi-symplectic form of the 1D gas dynamic equations. 
The theory implies that there are in general extra conservation laws 
that arise from using an expanded phase space involving the Clebsch variables. 
In 1D gas dynamics the symplecticity conservation law 
involves the space derivative of the energy conservation law, 
and the time derivative of the momentum conservation law. In this case and 
in more general cases (Sections 4 and 5), the extra symplecticity conservation laws 
imposes constraints that ensure that the equations involve the same effective number of 
dependent variables as the original Eulerian formulation of the fluid equations. 
The conserved densities and fluxes are written in terms of $2\times 2$ Jacobians of the 
dependent variables with respect to two of the independent space-time variables.} 

Two different formulations were investigated. In the first formulation, 
the constraint that the magnetic flux $\beta={\bf B}{\bf\cdot}d{\bf S}$ 
is Lie dragged
(i.e. conserved moving with the flow), leads to Faraday's equation in the form:
\beqn
\deriv{\bf B}{t}-\nabla\times({\bf u}\times{\bf B})
+{\bf u}(\nabla{\bf\cdot}{\bf B})=0, \label{eq:con1}
\eeqn
(for mathematical reasons it is useful to consider the 
case $\nabla{\bf\cdot}{\bf B}\neq 0$ as well as the 
physical case $\nabla{\bf\cdot}{\bf B}=0$). An alternative method to 
account for Faraday's equation, is to require that the gauge of the magnetic 
vector potential ${\bf A}$ is chosen so that 
the 1-form $\alpha={\bf A}{\bf\cdot}d{\bf x}$ is Lie dragged by the 
flow, i.e. ${\bf A}$ satisfies the evolution equation:
\beqn
\deriv{\bf A}{t}-{\bf u}\times(\nabla\times{\bf A})
+\nabla ({\bf u}{\bf\cdot}{\bf A})=0. \label{eq:con2}
\eeqn
For the case where ${\bf A}$ satisfies (\ref{eq:con2}), 
${\bf A}{\bf\cdot}{\bf B}/\rho$ is a scalar invariant advected by the 
flow, i.e.
\beqn
\frac{d}{dt}\left(\frac{{\bf A}{\bf\cdot}{\bf B}}{\rho}\right)=0, 
\label{eq:con3}
\eeqn
where $d/dt=\partial/\partial t+{\bf u}{\bf\cdot}\nabla$ is the Lagrangian 
time derivative following the flow and ${\bf B}=\nabla\times{\bf A}$. 

Following the approach of \cite{Cotter07}, we showed that the Clebsch 
variable evolution equations and the Clebsch representation for the 
the mass flux (momentum density) $\rho {\bf u}$ could be written in the 
multi-symplectic form:
\beqn
{\sf K}^\alpha_{ij}\deriv{z^j}{x^\alpha}=\frac{\delta{\cal H}}{\delta z^i}, 
\label{eq:con4}
\eeqn
where $x^\alpha$ ($\alpha=0,1,2,3$) denote the space-time coordinates 
$(t,x,y,z)$ and \\ 
${\bf z}=({\bf u}^T,\rho,S,
\mu,{\bf B}^T, \boldsymbol{\Gamma}^T, \lambda,\beta,\phi)^T$ is 
a 15-dimensional state vector describing the system. The covariant 
form of the multi-symplectic system (\ref{eq:con4}) 
for non-Cartesian spatial coordinates 
was discussed in Section 4 (see also \cite{Bridges10}).
The multi-symplectic Hamiltonian in (\ref{eq:con4}), given by
\beqn
{\cal H}=-\ell=-\int_V\left(\frac{1}{2}\rho |{\bf u}|^2-\varepsilon(\rho,S) -
\frac{B^2}{2\mu_0}\right)\ d^3x, \label{eq:con5}
\eeqn
is the negative of the MHD Lagrangian functional without constraints. The skew symmetric 
matrices ${\sf K}^\alpha_{ij}$ ($1\leq i,j\leq 15$) are related to 
one-forms $\omega^\alpha$, by the equations:
\beqn
\kappa^\alpha=\frac{1}{2} {\sf K}^\alpha_{ij}({\bf z}) dz^i\wedge dz^j
=d\omega^\alpha
\quad \hbox{and}\quad \omega^\alpha= L^\alpha_j dz^j, \label{eq:con6}
\eeqn
where the equation $\kappa^\alpha_{;\alpha}=0$ is the symplecticity 
 conservation law.

 Section 4  discusses multi-symplectic systems, 
based in part on the work of \cite{Hydon05}. This included a discussion of 
skew symmetric operators and matrices and Poisson brackets in 
Hamiltonian systems in which the time is the evolution variable. 
The conservation of the phase space element following the Hamiltonian flow 
and its  generalization for multi-symplectic 
systems (i.e. the symplecticity conservation law) were derived. 
Noether's theorem for 
multi-symplectic systems was discussed, and the form of 
the equations for generalized non-Cartesian space coordinates 
were studied (see \cite{Bridges10} for a study of multi-symplectic
systems and the variational bi-complex using the total exterior algebra tangent bundle (TEA)).   
 Proposition \ref{propforms}
shows that multi-symplectic systems can be written as a Cartan-Poincar\'e 
form equation using an $N+2$-form where $N$ is the number of 
independent space variables (see also \cite{Marsden99}), which  
involves the one-forms $\omega^\alpha$ 
and their exterior derivatives $d\omega^\alpha$ ($0\leq\alpha\leq N$) and the 
multi-symplectic Hamiltonian $H$. 
The Cartan-Poincar\'e form can be related to   
 Cartan's geometric formulation of partial differential equations 
(e.g. \cite{Harrison71}). 
Multi-symplectic systems for generalized (Cartesian and non-Cartesian) space 
coordinates were derived    
 from the  variational principle.    
 
In Section 5, we demonstrated that the multi-symplectic MHD equations 
obtained using the magnetic vector potential ${\bf A}$ 
satisfying (\ref{eq:con2}) are related to the multi-symplectic MHD 
equations using ${\bf B}$ and Faraday's law (\ref{eq:con1}) 
by the transformations (\ref{eq:ada5}), i.e.
\beqn
\boldsymbol{\Gamma}^*=-{\bf A} \quad \hbox{and}
\quad {\bf B}^*=\boldsymbol{\gamma}.  \label{eq:con7}
\eeqn
In (\ref{eq:con7}) $\psi^*$ is the image of $\psi$ under the map.
 ${\bf A}$ 
is the magnetic vector potential and $\boldsymbol{\gamma}$ is the corresponding
Lagrange multiplier in the advected ${\bf A}$ 
variational principle (\ref{eq:can2}). 
 ${\bf B}=\nabla\times{\bf A}$ and $\boldsymbol{\Gamma}$ 
is the  Lagrange multiplier for Faraday's equation in the
variational principle (\ref{eq:Clebsch1}). Under the map (\ref{eq:con7}),  
$\boldsymbol{\gamma}\to {\bf B}^*\equiv {\bf B}$ 
and ${\bf A}\to -\boldsymbol{\Gamma}^*\equiv-\boldsymbol{\Gamma}$.
{\bf In Section 5 we also obtained the 6 symplecticity conservation laws that 
occur when there are 4 independent (space-time) variables. These conservation 
laws are obtained from setting combinations of the derivatives of the momentum and energy 
conservation equations equal to zero, which ensures conservation of phase space. 
These conservation laws have densities $D$ and fluxes $F$ that consist of a sum 
of $2\times 2$ Jacobians of the dependent variables with respect to the space-time 
coordinates.}  
One can also derive conservation laws  
using the multi-symplectic versions of Noether's theorems. In particular, 
it is possible to derive the generalized non-local 
cross helicity conservation 
law for MHD and the generalized non-local 
helicity conservation law for ideal fluids, 
that apply for non-barotropic equations of state for the gas. These 
conservation laws were derived in \cite{Webb14b} (paper II) 
using Noether's theorem, fluid relabelling symmetries and gauge symmetries for the Lagrangian.  They depend on the nonlocal Clebsch variables.
For barotropic gases, these conservation laws reduce to the 
usual local 
cross helicity conservation law for MHD, 
and the helicity conservation law for ideal fluids.
However, because
 the Eulerian fluid velocity variation (Lie symmetry generator) 
for relabelling symmetries is zero, this implies a constraint on the 
Clebsch variable symmetry generators (see e.g. \cite{Calkin63}). 
Non-local conservation laws in partial differential equation systems 
 can arise from  Lie potential symmetries  
 of the cover system of equations, consisting of the 
the original system augmented by the differential equations for the 
potentials 
of the original system
 (e.g. \cite{Bluman10}, \cite{Sjoberg04}, 
\cite{Webb09}). This is worth further investigation. 
 
The Clebsch variable approach to MHD and fluid equations 
used in the present paper 
is not necessarily the only way that multi-symplectic systems of 
equations can be derived. \cite{Bridges06} in a study of elliptic 
partial differential equations (pdes) using the total exterior algebra bundle (TEA) 
makes the interesting 
observation that the symplectic matrices 
can sometimes have more obvious symmetry properties when higher order 
matrices are used to describe the system (e.g. the quaternion algebra 
is revealed when using $4\times 4$ matrices to describe 
the 2D Klein Gordon equation, which is not obvious when $3\times 3$ 
matrices are used to describe the multi-symplectic structure).  

\section*{Aknowledgements}
GMW acknowledges stimulating discussions of MHD conservation laws 
and multi-symplectic systems with Darryl Holm. {\bf We acknowledge discussions 
of the non-canonical MHD Poisson bracket and multi-symplectic MHD with 
Phillip Morrison.}
GPZ was supported
in part by NASA grants
NN05GG83G and NSF grant
nos. ATM-03-17509 and ATM-04-28880.
JFMcK acknowledges support by the NRF of South Africa.

\appendix
\section*{Appendix A}
\setcounter{section}{1}
In this appendix we derive the momentum and energy conservation equations
(\ref{eq:2.15a}) and  (\ref{eq:2.16a})  
for  one dimensional gas dynamics using the symplecticity pullback conservation 
laws (\ref{eq:m18}). We also derive the structural or simplecticity conservation 
law (\ref{eq:2.17a}) using (\ref{eq:m19}). 

The energy and momentum conservation laws (\ref{eq:2.15a}) and (\ref{eq:2.16a}) 
follow from (\ref{eq:m18}), i.e., 
\begin{equation}
D_\alpha\left(L^\alpha_j(z) z^j_{,\nu}-L\delta^\alpha_\nu\right)=0, \label{eq:A1}
\end{equation}
where the Lagrange density $L$ is given by (\ref{eq:2.9a}) and the one-forms 
$\omega^\alpha=L^\alpha_{z^s} dz^s$ are given by (\ref{eq:2.13a}) and $\alpha=0,1$ 
and $\nu=0,1$ for 1D gas dynamics and ${\bf z}=(u,\rho,S,\beta,\phi)^T$ are 
the dependent variables, $x^0=t$ and $x^1=x$. 
The fluid velocity 
$u$ is given by the Clebsch form (\ref{eq:2.1a}), i.e.
\begin{equation}
u=\deriv{\phi}{x}-\frac{\beta}{\rho} \deriv{S}{x}. \label{eq:A2}
\end{equation}

For $\nu=0$ (\ref{eq:A1}) gives the energy conservation equation:
\begin{equation}
\deriv{D_0}{t}+\deriv{F_0}{x}=0, \label{eq:A3}
\end{equation}
where
\begin{equation}
D_0=L^0_{z^j} z^j_{,0}-L,\quad F_0=L^1_{z^j} z^j_{,0}, \label{eq:A4}
\end{equation}
are the conserved density $D_0$ and flux $F_0$. The Lagrange density $L$ 
is given by (\ref{eq:2.9a})-(\ref{eq:2.10a}), i.e.:
\begin{equation}
L=\frac{1}{2}\rho u^2-\varepsilon(\rho,S)
+\phi\left[\deriv{\rho}{t}+\derv{x}(\rho u)\right]
+\beta \left(\deriv{S}{t}+u\deriv{S}{x}\right). \label{eq:A5}
\end{equation}
From (\ref{eq:A4}) we obtain:
\begin{equation}
D_0=\phi\rho_t+\beta S_t-L,\quad F_0=u\phi\rho_t+\beta u S_t+\phi\rho u_t. 
\label{eq:A6}
\end{equation}
Using $L$ from (\ref{eq:A5}) and using the Clebsch expansion (\ref{eq:A2}) 
for $u$ in (\ref{eq:A6}) we obtain:
\begin{align}
D_0=&\frac{1}{2}\rho u^2+\varepsilon (\rho,S)-D_x(\rho u\phi), \nonumber\\
F_0=&D_t(\rho u\phi)+\rho u\left(\frac{1}{2} u^2+h\right). \label{eq:A7}
\end{align}
Substitution of $D_0$ and $F_0$ from (\ref{eq:A7}) in (\ref{eq:A3}) 
gives the energy conservation law (\ref{eq:2.15a}). 

The momentum conservation law from (\ref{eq:A1}) has the form:
$D_t(D_1)+D_x(F_1)=0$ where
\begin{equation}
D_1=L^0_{z^j} z^j_{,1},\quad F_1=L^1_{z^j} z^j_{,1}-L. \label{eq:A8}
\end{equation}
We find:
\begin{align}
D_1=&\phi\rho_x+\beta S_x\equiv D_x(\rho\phi)-\rho u, \nonumber\\
F_1=&u\phi \rho_x+u\beta S_x+\phi\rho u_x-L\nonumber\\
=&-\left(\rho u^2+p\right)-D_t(\rho\phi)-\beta(S_t+uS_x). \label{eq:A9}
\end{align}
From (\ref{eq:A8})-(\ref{eq:A9}) we obtain:
\begin{equation}
\deriv{D_1}{t}+\deriv{F_1}{x}=-\left[\derv{t}(\rho u)+\derv{x}\left(\rho u^2+p\right)\right]=0, 
\label{eq:A10}
\end{equation}
which is the momentum conservation equation (\ref{eq:2.16a}). 

In general, the symplecticity conservation laws are given by (\ref{eq:m19}), i.e. 
\begin{equation}
D_\alpha\left(F^\alpha_{\nu\gamma}\right)=0,\quad\hbox{where}\quad \nu<\gamma, 
\label{eq:A11}
\end{equation}
and
\begin{equation}
F^\alpha_{\nu\gamma}={\sf K}^\alpha_{ij} z^i_{,\nu} z^j_{,\gamma}. \label{eq:A12}
\end{equation}
For the case of 1D gas dynamics, there is only one structural conservation law, namely:
\begin{equation}
D_t\left(F^0_{01}\right)+D_x\left(F^1_{01}\right)=0. \label{eq:A13}
\end{equation}
Using (\ref{eq:A12}) we find:
\begin{align}
F^0_{01}=& z^i_{,0} {\sf K}^0_{ij} z^j_{,1}\equiv [{\bf z}_t]^T{\sf K}^0 {\bf z}_x\nonumber\\
=&\phi_t\rho_x+\beta_tS_x-S_t\beta_x-\phi_x\rho_t\nonumber\\
=&\frac{\partial(\phi,\rho)}{\partial(t,x)} +\frac{\partial(\beta,S)}{\partial(t,x)}. \label{eq:A14}
\end{align}
Similarly we find:
\begin{align}
F^1_{01}=&z^i_{,0} {\sf K}^1_{ij} z^j_{,1}\equiv [{\bf z}_t]^T {\sf K}^1 {\bf z}_x\nonumber\\
=&\left(\rho\phi_t-\beta S_t\right) u_x+u\phi_t\rho_x+\derv{t}(\beta\rho) S_x
-u S_t\beta_x-\phi_x\derv{t}(\rho u)\nonumber\\
=&\rho \frac{\partial(\phi,u)}{\partial(t,x)}+ u \frac{\partial(\phi,\rho)}{\partial(t,x)} 
+\beta \frac{\partial(u,S)}{\partial(t,x)}+ u \frac{\partial(\beta,S)}{\partial(t,x)}, 
\label{eq:A15}
\end{align}
Using (\ref{eq:A14}) and (\ref{eq:A15}) in (\ref{eq:A13}) gives the symplecticity 
conservation law (\ref{eq:2.17a}).

\section*{Appendix B}
\setcounter{section}{2}
In this appendix we discuss the notation 
\beqn
V_{\bf B}={\bf B}{\bf\cdot}\nabla\circ-{\bf B}{\bf\cdot}(\nabla\circ)^T, 
\quad V_{\bf B}^\dagger ={\bf B}{\bf\cdot}(\nabla\circ)^T
-{\bf B}(\nabla{\bf\cdot}\circ)
, \label{eq:B1}
\eeqn
used in (\ref{eq:m34}). Consider the integral:
\begin{align}
\int_{R}V_{\bf B}({\bf W})\ d^3x=&\int_R\left[{\bf B}{\bf\cdot}\nabla\circ
-{\bf B}{\bf\cdot}(\nabla\circ)^T\right]{\bf W}\ d^3x\nonumber\\
=&\int_R\left(B^s\deriv{W^i}{x^s}-B^s\nabla_i W^s\right){\bf e}_i\ d^3x
\nonumber\\
=&\int_R\left\{\nabla{\bf\cdot}\left({\bf B}\otimes{\bf W}\right)
-\nabla({\bf B}{\bf\cdot}{\bf W})-{\bf W}(\nabla{\bf\cdot}{\bf B})
+{\bf W}{\bf\cdot}(\nabla{\bf B})^T\right\} d^3x\nonumber\\
=&\int_R\left\{{\bf W}{\bf\cdot}(\nabla\circ)^T
-{\bf W}\nabla{\bf\cdot}\circ\right\}
{\bf B}\ d^3x
=\int_R V_{\bf W}^\dagger({\bf B})\ d^3x, \label{eq:B2}
\end{align}
where
\beqn
V_{\bf W}^\dagger ={\bf W}{\bf\cdot}(\nabla\circ)^T-{\bf W}(\nabla
{\bf\cdot}\circ)
, \label{eq:B3}
\eeqn
and the $\{{\bf e}_i\}$ are unit base vectors along the $x$, $y$ and $z$ axes. 
In the derivation of (\ref{eq:B2}) we have used Gauss' divergence 
theorem to obtain:
\begin{align}
\int_R\nabla{\bf\cdot}\left({\bf B}\otimes{\bf W}\right)\ d^3x
=&\int_{\partial R}
\left({\bf B}{\bf\cdot}{\bf n}\right){\bf W}\ d{\bf S},\nonumber\\
\int_R\nabla({\bf B}{\bf\cdot}{\bf W})\ d^3x=&
\int_{\partial R} ({\bf B}{\bf\cdot}{\bf W}){\bf n}\ dS, \label{eq:B4}
\end{align}
where ${\bf n}$ is the outward normal to the region $R$, and we assume that the surface integrals (\ref{eq:B4}) vanish. The formula $V_{\bf B}^\dagger$ in 
(\ref{eq:B1}) now follows by using the replacement 
${\bf W}\to{\bf B}$ in (\ref{eq:B3}).


\begin{thebibliography}{}




\bibitem[{\it Akhmetiev  and  Ruzmaikin} (1995)]{Akhmetiev95}
Akhmetiev, P., and Ruzmaikin, A. 1995, A fourth order topological invariant 
of magnetic or vortex lines, {\it J. Geom. Phys}. {\bf 15}, 95-101.

\bibitem[{\it Anderson}(1989)]{Anderson89}
Anderson, I. M. 1989, The Variational Bicomplex, book manuscript; 
Utah State University (1989). $http://www.math.usu.edu/~fg_mp/Publications/VB/vb.pdf$

\bibitem[{\it Anderson}(1992)]{Anderson92}
Anderson, I. M. 1992, Introduction to the variational Bi-complex. in Mathematical aspects of 
contemporary field theory, {\it Contemp. Math.}, {\bf 132}, 51-73.

\bibitem[{\it Arnold and Khesin}(1998)]{Arnold98}
Arnold, V. I. and Khesin, B. A. 1998, Topological Methods in Hydrodynamics, 
Springer, New York.



\bibitem[{\it Berger}(1990)]{Berger90}
Berger, M. A. 1990, Third -order link integrals, {\it J. Phys. A: Math. Gen.}, 
{\bf 23}, 2787-2793.

\bibitem[{\it Berger and Field}(1984)]{Berger84}
Berger, M. A. and Field, G. B., 1984,
The toplological properties of magnetic helicity, {\it J. Fluid. Mech.},
 {\bf 147}, 133-148.





\bibitem[{\it Bluman and Kumei}(1989)]{Bluman89}
Bluman, G. W. and Kumei, S. 1989, Symmetries and Differential Equations, 
(New York: Springer). 

\bibitem[{\it Bluman et al.}(2010)]{Bluman10}
 Bluman, G. W., Cheviakov, A. F. and Anco, S. 2010, Applications 
of Symmetry Methods to Partial Differential Equations, Springer: New York. 


\bibitem[{\it Bridges}(1992)]{Bridges92}
Bridges, T. J. 1992, Spatial Hamiltonian structure, energy flux and 
the water-wave problem.
{\it Proc. Roy. Soc. London},{\bf 439}, 297-315.

\bibitem[{\it Bridges}(1997a)]{Bridges97a}
 Bridges, T.J.,  1997a, Multi-symplectic structures and wave propagation, {\it 
Math. Proc. Camb. Philos. Soc.}, {\bf 121}, 147-190.

\bibitem[{\it Bridges}(1997b)]{Bridges97b} 
Bridges, T. J., 1997b,
A geometric formulation of the conservation of wave action and its 
implications for signature and classification of instabilities, {\it 
Proc. Roy. Soc.} A, {\bf 453}, 1365-1395 (1997b).

\bibitem[{\it Bridges et al.}(2005)]{Bridges05} 
Bridges, T. J. Hydon, P. E. and 
 Reich, S. 2005, Vorticity and
symplecticity in Lagrangian fluid dynamics, \emph{J. Phys. A: Math. Gen.} 
\textbf{38} 1403-1418 (2005)

\bibitem[{\it Bridges}(2006)]{Bridges06}
Bridges, T. J. and Reich, S., 2006, Numerical methods for Hamiltonian PDEs, 
{\it J. Phys. A, Math. Gen.}, {\bf 39}, 5287-5320.

\bibitem[{\em Bridges}(2006)]{Bridges06b}
Bridges, T. J. 2006, Canonical multi-symplectic structure on 
the total exterior algebra bundle, {\it Proc. Roy. Soc. London, A}, {\bf 462}, 
1531-1551.

\bibitem[{\em Bridges et al.}(2010)]{Bridges10}
Bridges, T. J., Hydon, P. E. and Lawson, J.K. 2010, Multi-symplectic structures and the variational bi-complex, {\it Math. Proc. Cambridge Phil. Soc.}, 
issue 1, (Jan. 2010), pp 159-178.

\bibitem[{\em Brio et al.}(2010)]{Brio10}
Brio, M., Zakharian, A. R. and Webb, G. M. 2010, Numerical Time-Dependent 
Partial Differential Equations for Scientists and Engineers, Mathematics in science and Engineering, 
vol.123, Elsevier Press, (Ed. C. K. Chui) 2010, 
first edition, pp. 199-204. 

\bibitem[{\it Calkin}(1963)]{Calkin63}
Calkin, M. G. 1963, An action principle for magnetohydrodynamics, {\it Canad.J. Phys.}, {\bf 41}, 2241-2251. 

\bibitem[{\it Chandre et al.}(2013)]{Chandre13}
Chandre, C., de Guillebon, L., Back, A., Tassi, E. and Morrison, P. J. 2013,
On the use of projectors for Hamiltonian systems and their relationship
with Dirac brackets, {\it J. Phys. A, Math. and theoret.}, {\bf 46},
125203 (14pp), doi:10.10.1088/1751-8133/46/12/125203.

\bibitem[{\it Cotter et al. }(2007)]{Cotter07}
Cotter, C. J., Holm, D. D., and Hydon, P. E., 2007, Multi-symplectic 
formulation of fluid dynamics using the inverse map, 
{\it Proc. Roy. Soc. Lond.} A, {\bf 463}, 2617-2687 (2007).









\bibitem[{\it Finn and Antonsen}(1985)]{Finn85}
Finn, J. H. and Antonsen, T. M. 1985,  Magnetic helicity: what is it and what is it good for?, {\it Comments on Plasma Phys. and Contr.
Fusion}, {\bf 9}(3), 111.

\bibitem[{\it Finn and Antonsen}(1988)]{Finn88}
Finn, J. M. and Antonsen, T. M. 1988, Magnetic helicity injection for
configurations with field errors, {\it Phys. Fluids}, {\bf 31} (10), 3012-3017.






\bibitem[{\it Gordin and Petviashvili}(1987)]{Gordin87}
Gordin, V. A. and Petviashvili, V. I., 1987, The gauge of vector potential and
Lyapunov stable MHD equilibrium, {\it Soviet J. Plasma Phys.}, 
{\bf 13}, No. 7,pp. 509-511 (English). 
 




\bibitem[{\it Harrison and Estabrook}(1971)]{Harrison71}
Harrison, B. K. and Estabrook, F. B., 1971, Geometric approach to invariance 
groups and solution of partial differential systems, {\it J. Math. Phys.},
 {\bf 12}, 653-66.








\bibitem[{\it Holm and Kupershmidt}(1983a)]{Holm83a}
Holm, D. D. and Kupershmidt, B. A. 1983a, Poisson brackets and Clebsch 
representations for magnetohydrodynamics, multi-fluid plasmas and 
elasticity, {\it Physica D}, {\bf 6D}, 347-363.

\bibitem[{\it Holm and Kupershmidt}(1983b)]{Holm83b}
Holm, D. D. and Kupershmidt, B. A. 1983b, Noncanonical Hamiltonian 
formulation of ideal magnetohydrodynamics, {\it Physica D}, 
{\bf 7D}, 330-333.

\bibitem[{\it Holm et al.}(1998)]{Holm98}
Holm, D.D., Marsden, J.E. and Ratiu, T.S. 1998, The Euler-Lagrange equations and semi-products with application to continuum theories, {\it Adv. Math.},
{\bf 137}, 1-81.


\bibitem[{\it Hydon}(2005)]{Hydon05} 
Hydon, P. E., 2005, Multi-symplectic conservation laws for differential 
and differential-difference equations,
\emph{Proc. Roy. Soc. A}, \textbf{461}, 1627-1637 (2005).

\bibitem[{\it Hydon and Mansfield}(2011)]{Hydon11}
Hydon, P. E. and Mansfield, E. L. 2011, Extensions of Noether's second theorem: from continuous to discrete systems, {\it Proc. Roy. Soc. A}, 
{\bf 467}, pp. 3206-3221, doi:10.1098/rspa.2011.0158.










\bibitem[{\it Marsden and Ratiu} (1994)]{Marsden94}
Marsden J. E. and Ratiu T.S. 1994, Introduction to Mechanics and Symmetry,
New York,: Springer Verlag.

\bibitem[{\it Marsden and Shkoller}(1999)]{Marsden99}
Marsden, J. E. and Shkoller, S. 1999, Multi-symplectic geometry, covariant
Hamiltonians and Water Waves, {\it Math. Proc. Camb. Phil. Soc.}, 
{\bf 125}, 553-575.



\bibitem[{\it Moffatt}(1969)]{Moffatt69}
Moffatt, H. K. 1969, The degree of knottedness of tangled vortex lines,
{\it J. Fluid. Mech.}, {\bf 35}, 117.




\bibitem[{\it Morrison}(1982)]{Morrison82}
Morrison, P. J. 1982, Poisson brackets for fluids and plasmas, in Mathematical 
Methods in Hydrodynamics and Integrability of dynamical Systems,
 {\it AIP Proc. Conf.}, {\bf 88}, ed M. Tabor and Y. M. Treve, pp 13-46.

\bibitem[{\it Morrison and Greene}(1980)]{Morrison80}
Morrison, P.J. and Greene, J.M. 1980, Noncanonical Hamiltonian density
formulation of hydrodynamics and ideal magnetohydrodynamics,
 {\it Phys. Rev. Lett.}, {\bf 45}, 790-794.

\bibitem[{\it Morrison and Greene}(1982)]{Morrison82a}
Morrison, P.J. and Greene, J.M. 1982, Noncanonical Hamiltonian density
formulation of hydrodynamics and ideal magnetohydrodynamics, (Errata),
 {\it Phys. Rev. Lett.}, {\bf 48}, 569.

\bibitem[{\it Morrison}(1998)]{Morrison98}
Morrison, P. J. 1998, 
Hamiltonian description of the ideal fluid, {\it Reviews of Modern Physics}, 
{\bf 70}, Issue 2, April 1998, pp.467-521, 	
10.1103/RevModPhys.70.467

\bibitem[{\it Newcomb}(1962)]{Newcomb62}
Newcomb, W. A. 1962, Lagrangian and Hamiltonian methods 
in magnetohydrodynamics, {\it Nucl. Fusion Suppl.}, Part 2, 451-463. 




\bibitem[{\it Padhye and Morrison}(1996a)]{Padhye96a}
Padhye, N. S. and Morrison, P. J. 1996a, Fluid relabeling symmetry,
 {\it Phys. Lett. A}, {\bf 219}, 287-292.

\bibitem[{\it Padhye and Morrison}(1996b)]{Padhye96b}
Padhye, N. S. and Morrison, P. J. 1996b, Relabeling symmetries in hydrodynamics 
and magnetohydrodynamics, {\it Plasma Phys. Reports}, {\bf 22}, 869-877.


\bibitem[{\it Powell et al.}(1999)]{Powell99}
Powell, K. G., Roe, P.L., Linde, T.J., Gombosi, T. I., and De Zeeuw, D.
1999, A solution adaptive upwind scheme for ideal magnetohydrodynamics,
{J. Comput. Phys.} {\bf 154}, 284-309.

\bibitem[{\it Reich}(2000)]{Reich00}
Reich, S. 2000, Multi-symplectic Runge-Kutta Collocation methods for Hamiltonian wave equations, {\it J. Comp. Phys.}, {\bf 57}, 473.








\bibitem[{\it Ruzmaikin and Akhmetiev} (1994)]{Ruzmaikin94}
Ruzmaikin, A. and Akhmetiev, P. 1994, Topological invariants of magnetic fields and the effect of reconnections,  {\it Phys. Plasmas}, {\bf 1}, No. 2,  
331-338.

\bibitem[{\it Sj\"oberg and Mahomed,}(2004)]{Sjoberg04}
Sj\"oberg, A. and Mahomed, F.M. 2004, Non-local symmetries and 
conservation laws for one-dimensional gas dynamics equations,
 {\it Appl. Math. Comput.}, 
{\bf 150}, 379-397.

\bibitem[{\it Tur and Yanovsky}(1993)]{Tur93}
Tur, A. V. and Yanovsky, V. V. 1993, Invariants in dissipationless 
hydrodynamic media, {\it J. Fluid Mech.}, {\bf 248}, 
Cambridge Univ. Press, p67-106.






\bibitem[{\it Webb et al.}(2005)]{Webb05b}
Webb, G. M., Zank, G.P., Kaghashvili, E. Kh and Ratkiewicz, R.E., 2005,
Magnetohydrodynamic waves in non-uniform flows II: stress energy tensors, 
conservation laws and Lie symmetries,
 {\it J. Plasma Phys.}, {\bf 71}, 811-857, doi: 10.1017/S00223778050003740.



\bibitem[{\it Webb et al.}(2007)]{Webb07}
Webb, G. M.; McKenzie, J. F.; Mace, R. L.; Ko, C. M.; Zank, G. P. 2007,
Dual variational principles for nonlinear traveling waves in multifluid plasmas,
Phys. of Plasmas, {\bf 4}, Issue 8, pp. 082318-082318-17, doi:10.1063/1.2757154

\bibitem[{\it Webb et al.} (2008)]{Webb08} 
Webb, G. M., Ko, C. M., Mace, R.L., McKenzie, J.F. and Zank, G.P. 2008,
Integrable, oblique travelling waves in charge neutral, two-fluid plasmas,
{\it Nonl. Proc. Geophys.}, {\bf 15}, 179-208.

\bibitem[{\it Webb and Zank}(2009)]{Webb09}
Webb, G. M. and Zank, G. P. 2009, Scaling symmetries, conservation laws and 
action principles in one-dimensional gas dynamics, 
{\it J. Phys. A., Math. Theor.}, {\bf 42}, 475205 (23pp). 

\bibitem[{\it Webb et al.} (2010)]{WebbPogorelovZank10}
Webb, G.M., Pogorelov, N.V. and Zank, G.P. 2010, MHD simple waves and
the divergence wave, {\em Solar Wind}, {\bf 12}, St. Malo, France,
 {\it AIP Proc. Conf.}, {\bf 1216}, pp.300-303. doi:10.1063/1.3396300.



\bibitem[{ \it Webb et al.}(2014a)]{Webb14a}
Webb, G. M., Dasgupta, B., McKenzie, J. F., Hu, Q.,
and Zank, G.P. 2014a:
Local and nonlocal advected invariants and helicities in
magnetohydrodynamics and gas dynamics I: Lie dragging approach,
{\it J. Phys. A., Math. and Theoret.}, {\bf 47}, (2014) 095501 (33pp).
doi:10.1088/1751-8113/49/9/095501, preprint 
available at http://arxiv.org/abs/1307.1105.

\bibitem[{\it Webb et al.} (2014b)]{Webb14b}
Webb, G. M., Dasgupta, B., McKenzie, J. F., Hu, Q.,
and Zank, G.P. 2014b:
Local and nonlocal advected invariants and helicities in
magnetohydrodynamics and gas dynamics II: Noether's theorems and Casimirs,
{\it J. Phys. A., Math. and Theoret.}, {\bf 47} (2014) 095502 (31pp), 
doi:10.1088/1751-8113/47/9/095502, preprint  
available at http://arxiv.org/abs/1307.1038. 

\bibitem[{\it Webb et al.}(2014c)]{Webb14c}
Webb, G. M., Hu, Q., McKenzie, J. F., Dasgupta, B. and Zank, G.P. 2014c, 
Advected invariants
in MHD and gas dynamics, {\it 12th Ann. Internat. Astrophys. Conf.}, 
in Outstanding Problems in Heliophysics: from Coronal Heating to
the Edge of the Heliosphere, eds. G.P. Zank and Q. Hu, 
{\it Astronomical Society of the Pacific
Conf. Series}, {\bf 484}, 229-234.

\bibitem[{\it Webb et al.} (2014d)]{Webb14d}
Webb, G. M., Burrows, R. H., Ao, X., and Zank, G.P. 2014d, Ion acoustic 
travelling waves, {\it J. Plasma Phys.}, {\bf 80}, part 2, pp. 147-171,
doi:10.1017/S0022377813001013, 
preprint at http://arxiv.org/abs/1312.6406.

\bibitem[{\it Webb and Mace}(2014)]{WebbMace14}
Webb, G. M. and Mace, R.L. 2014, Noether's Theorems and fluid relabelling 
symmetries in magnetohydrodynamics and gas dynamics, 
{\it J. Phys. A, Math. and Theoret.}, 
article JPHYSA-101-057, submitted March 10, 2014, available
at http://arxiv.org/abs/1403.3133





\bibitem[{\it Woltjer}(1958)]{Woltjer58}
Woltjer, L., On hydromagnetic equilibria, {\it Proc. Nat. Acad. of Sciences},
{\bf 44}, No. 9, 833-841.



\bibitem[{\it Zakharov and Kuznetsov}(1997)]{Zakharov97}
Zakharov, V. E. and Kuznetsov, E.A. 1997, Hamiltonian formalism for 
nonlinear waves, {\it Physics-Uspekhi}, {\bf 40}, (11), 1087-1116.
\end{thebibliography}
\end{document}